\documentclass[11pt]{article}
\pdfoutput=1

\usepackage[letterpaper,margin=1in]{geometry} 
\usepackage{setspace} 
\usepackage{lmodern} 

\usepackage[T1]{fontenc}
\usepackage[utf8]{inputenc}
\usepackage{microtype}

\usepackage{amsmath}
\usepackage{amssymb}
\usepackage{amsthm}
\usepackage{graphicx}
\usepackage{mathrsfs}
\usepackage{braket}
\usepackage{bbm}
\usepackage[usenames,dvipsnames]{xcolor}
\usepackage{cite}
\usepackage{doi}
\usepackage{tikz}
\usepackage{physics}
\usepackage{booktabs}
\usepackage{threeparttable}
\usepackage{multirow}

\usetikzlibrary{quantikz}
\usepackage{theoremref}

\usepackage{algorithm}
\usepackage{algorithmic}

\newtheorem{thm}{\protect\theoremname}
\theoremstyle{plain}
\newtheorem{lem}{\protect\lemmaname}
\theoremstyle{plain}
\newtheorem{rem}[thm]{\protect\remarkname}
\theoremstyle{plain}
\newtheorem*{lem*}{\protect\lemmaname}
\theoremstyle{plain}
\newtheorem*{thm*}{\protect\theoremname}
\theoremstyle{plain}

\theoremstyle{plain}
\newtheorem{cor}[thm]{\protect\corollaryname}

\newtheorem{defn}{Definition}

\usepackage[USenglish]{babel}
  \providecommand{\corollaryname}{Corollary}
  \providecommand{\lemmaname}{Lemma}
  \providecommand{\propositionname}{Proposition}
  \providecommand{\remarkname}{Remark}
\providecommand{\theoremname}{Theorem}

\usepackage[normalem]{ulem}

\usepackage{authblk}

\hypersetup{
    colorlinks = true, 
    linkcolor = MidnightBlue, 
    urlcolor  = MidnightBlue, 
    citecolor = MidnightBlue,
}

\newcommand{\CC}{\mathbb{C}}

\newcommand{\coef}{\mu}
\newcommand{\eigen}{\lambda}

\renewcommand{\Im}{\operatorname{Im}}

\title{Learning $k$-body Hamiltonians via compressed sensing}
\author[1,6]{Muzhou Ma}
\author[2,3]{Steven T. Flammia}
\author[1,7]{John Preskill}
\author[1,4,5]{Yu Tong}
\affil[1]{Institute for Quantum Information and Matter, California Institute of Technology, CA 91125, USA}
\affil[2]{Department of Computer Science, Virginia Tech, Alexandria, VA 22314, USA}
\affil[3]{Phasecraft Inc., Washington DC 20001, USA}
\affil[4]{Department of Mathematics, Duke University, Durham, NC 27708, USA}
\affil[5]{Department of Electrical and Computer Engineering, Duke University, Durham, NC 27708, USA}
\affil[6]{Department of Electronic Engineering, Tsinghua University, Beijing, China}
\affil[7]{AWS Center for Quantum Computing, Pasadena, CA 91125, USA}
\date{\today}

\begin{document}

\maketitle

\begin{abstract}
    We study the problem of learning a $k$-body Hamiltonian with $M$ unknown Pauli terms that are not necessarily geometrically local. We propose a protocol that learns the Hamiltonian to precision $\epsilon$ with total evolution time ${\mathcal{O}}(M^{1/2+1/p}/\epsilon)$ up to logarithmic factors, where the error is quantified by the $\ell^p$-distance between Pauli coefficients. Our learning protocol uses only single-qubit control operations and a GHZ state initial state, is non-adaptive, is robust against SPAM errors, and performs well even if $M$ and $k$ are not precisely known in advance or if the Hamiltonian is not exactly $M$-sparse. Methods from the classical theory of compressed sensing are used for efficiently identifying the $M$ terms in the Hamiltonian from among all possible $k$-body Pauli operators. We also provide a lower bound on the total evolution time needed in this learning task, and we discuss the operational interpretations of the $\ell^1$ and $\ell^2$ error metrics. In contrast to most previous works, our learning protocol requires neither geometric locality nor any other relaxed locality conditions. 
\end{abstract}
\clearpage
\tableofcontents
\clearpage

\section{Introduction}
\label{sec:introduction}
Hamiltonian learning is a fundamental problem in quantum physics, crucial for understanding and controlling quantum systems. 
Accurate models of Hamiltonians are needed for tasks ranging from quantum simulation to quantum error correction, where precise knowledge of interactions allows for better optimization and control. 
It is also a fundamental problem in computer science, where the classically analogous task of learning undirected graphical models is central in many models of machine learning. 

Learning Hamiltonians, particularly in large systems with many-body interactions, is complicated by the fact that dynamical access to the Hamiltonian naturally leads to entangled quantum states. 
Inferring a general unknown Hamiltonian from measurements in the computational basis is therefore nontrivial. 
Similarly, access via Gibbs states introduces challenges in state preparation and post-processing of measurement results. 
Prior work, which we review below, has nonetheless shown that an unknown Hamiltonian can be learned in polynomial time from its dynamics \cite{HaahKothariTang2022optimal,StilckFrança2024,yu2023robust,gu2022practical,HuangTongFangSu2023learning,Caro_2024} or from samples of a Gibbs state when the bipartite graph connecting qubits to interactions has bounded degree~\cite{anshu2021sample,HaahKothariTang2022optimal,BakshiLiuMoitraTang2024learning}. 

The classical theory of compressed sensing~\cite{FoucartRauhut2013, candes2008introduction, chen2001atomic, candes2006near, rudelson2008sparse, candes2006stable, candes2011probabilistic, eldar2012compressed} is a promising avenue for improving Hamiltonian learning and extending it to the setting where the interactions are sparse, or approximately sparse, in a known basis.
It also has many applications in other types of learning tasks involving quantum states or processes \cite{Seif2021compressed, gross2010quantum, shabani2011efficient, shabani2011estimation, flammia2012quantum, Kalev2015, roth2018recovering, harper2021fast}.
Classically, compressed sensing efficiently solves the following problem: given a linear measurement matrix $A :\mathbb{C}^{n} \to \mathbb{C}^m$ and an unknown vector $x \in \mathbb{C}^n$ which is promised to be $s$-sparse in a given basis, reconstruct $x$ from $y = Ax$ when $m = \mathcal{O}(s \log n)$. 
At first glance this seems impossible for at least two reasons. First, the well-known Nyquist theorem states that learning a signal with bandwidth $n$ requires measuring $\Omega(n)$ frequencies. 
This however does not take into account the sparsity of the underlying signal. 
But second, naively taking the sparsity into account quickly runs into fundamental computational bottlenecks, as finding the sparsest vector consistent with affine constraints is $\mathsf{NP}$-hard \cite{Natarajan1995sparse}. 

Under surprisingly mild conditions on the measurement matrix $A$, compressed sensing overcomes these two obstacles. 
Moreover, the compressed-sensing estimate of the signal is robust to measurement noise due to sampling, robust to the signal being only ``near'' to a sparse signal, and is even efficient to compute via a reduction to an easy convex optimization problem. 
It is therefore natural to hope that compressed sensing can offer a powerful approach to reconstructing sparse Hamiltonians. 

\paragraph{Main results.}
We consider the class of $n$-qubit $k$-body Hamiltonians, where all couplings are written in the Pauli basis and involve Paulis of weight at most $k$ with $k = \mathcal{O}(1)$. 
Let $M \le \sum_{l=0}^k 3^l\binom{n}{l} = \mathcal{O}(n^k)$ be the total number of nonzero couplings in the Hamiltonian. 
Write $H = \sum_{a=1}^M \coef_a P_a$ for the explicit unknown Hamiltonian, where the coefficients $\coef_a$ are assumed to be normalized to lie in the interval $[-1,1]$. 
Our goal is to output the unknown Paulis $P_a$ and estimates $\widehat{\coef}_a$ of the coefficients that are accurate in $\ell^p$ error, defined as $\bigl\|\widehat{\boldsymbol{\coef}}-\boldsymbol{\coef}\bigr\|_p = \bigl(\sum_{a=1}^M |\widehat{\coef}_a - \coef_a|^{p}\bigr)^{1/p}$ for suitable $p$. 

In what follows, we use $\widetilde{\mathcal{O}}(f)$ to represent $\mathcal{O}(f\mathrm{polylog}(n,M,1/\delta,1/\epsilon,2^k))$ where $\delta$ is the failure probability and $\epsilon$ is the precision that we want to achieve.

\begin{thm*}[Informal version of Theorem~\ref{thm:ham_learning_upper_bound}]
    There exists a learning protocol that uses $N=\widetilde{\mathcal{O}}(M)$ independent non-adaptive experiments to learn, with probability at least $1-\delta$, an estimate $\widehat{\boldsymbol{\coef}}$ of all coefficients with weight $\le k$ that has $\ell^p$ error ($1\leq p\leq 2$) at most $\epsilon$, and uses a total evolution time
    \[
        T =\widetilde{\mathcal{O}}\left(\frac{M^{1/p+1/2}}{\epsilon}\right).
    \]
    The protocol is robust to a constant amount of state preparation and measurement (SPAM) noise. 
\end{thm*}
In the above, by \emph{total evolution time}, we mean the total time needed to evolve with the unknown Hamiltonian in all experiments.
Our method is also robust to modeling errors (i.e., the Hamiltonian may not be exactly $k$-local or have more than $M$ terms), as discussed in Section~\ref{sec:robustness_to_modeling_errs}.
The operational interpretations of the $\ell^p$ error metric for $p=1$ and $2$ are discussed in Section~\ref{sec:operational_interpretation} and correspond to notions of worst- and average-case error, respectively.
It is notable that our algorithm nearly achieves Heisenberg-limited scaling, ${\mathcal{O}}(1/\epsilon)$, in the total evolution time. 
Also importantly, the classical computation required by our algorithm is polynomial. 
In fact, our estimator reduces to a convex optimization of the form
\begin{equation}
\label{eq:exabsopt}\operatornamewithlimits{minimize}_{\boldsymbol{x}}\|\boldsymbol{x}\|_1 \quad \text{ subject to }\ \|\boldsymbol{A} \boldsymbol{x}-\boldsymbol{y}\|_2 \leq \alpha.
\end{equation}
In Theorem~\ref{thm:compressed_sensing_approx_solution}, we give an explicit bound on the computational resources required to obtain an (approximately) optimal solution to Eq.~(\ref{eq:exabsopt}). 
As our analysis of this optimization is self-contained (it does not rely on reducing Eq.~\ref{eq:exabsopt} to a semidefinite program) and appears to be new, it may be of interest to a wider community. 

We also prove a lower bound for the total evolution time for the case of $\ell^1$ error. 
\begin{thm*}[Informal version of Theorem~\ref{thm:lower_bound}]
    Suppose a learning algorithm estimates $\widehat{\boldsymbol{\coef}}$ for the unknown Hamiltonian $H=\sum_{a=1}^M \coef_aP_a$ with $\|\widehat{\boldsymbol{\coef}}-\boldsymbol{\coef}\|_1 \le \epsilon$ in expectation value over experimental outcomes. 
    Then 
    \begin{equation}
        T = {\Omega}\biggl(\frac{M}{\epsilon\log(1/\gamma)}\biggr).
    \end{equation}
    Here $\gamma$ is the amount of measurement error (see \eqref{eq:global_depolarizing}).
\end{thm*}
Interestingly, our lower bound holds for learning $n$-qubit Hamiltonians with the $M$ terms \textit{fixed} and having \textit{known} support. 
This naturally implies a lower bound for learning $n$-qubit $M$-term Hamiltonians without prior knowledge of which $M$ terms are in the Hamiltonian. 
This suggests that parameter estimation, rather than support identification, may be the primary obstacle to improving the time complexity for Hamiltonian learning.

\paragraph{Techniques.}
The lack of locality presents a considerable challenge to the Hamiltonian learning problem, as it prevents us from decoupling the learning problem into small problems that can be efficiently solved.
We are therefore compelled to deal with the entire quantum system as a whole, and new techniques are needed to circumvent the exponential cost that usually comes from such problems.
Compressed sensing utilizes sparsity to greatly reduce the complexity of reconstructing a high-dimensional vector, but prior to this work, it was far from clear how it can be applied to Hamiltonian learning. 
Although the Hamiltonian coefficients can be collected into a sparse vector, a generic observable does not depend on these coefficients linearly. 
The non-commutativity of the terms presents additional difficulties.

One key observation that led to this work is that we can apply compressed sensing to learn a completely commuting Hamiltonian, i.e., a Hamiltonian whose Pauli terms commute \emph{on each qubit}. 
For such a Hamiltonian, the eigenvalue differences depend linearly on the coefficients, and we can estimate these differences efficiently as the eigenstates are all product states.

We then extend this method to general $k$-body Hamiltonians with $M$ terms using Hamiltonian reshaping. 
Hamiltonian reshaping was used to decouple the dynamics and choose local bases in \cite{HuangTongFangSu2023learning}, but here we use it to obtain an effective Hamiltonian that is diagonal in a global basis set of our choosing. 
This effective Hamiltonian can then be learned using compressed sensing.
However, this procedure will only allow us to learn the terms that are diagonal in the given basis, and multiple choices of bases are needed to cover all terms in the Hamiltonian. 
We then choose the basis randomly, and this allows us to cover all $k$-body terms with mild overhead. 
This idea comes from randomized measurement \cite{evans2019scalablebayesianhamiltonianlearning,huang2020predicting,Elben_2022}, but to the best of our knowledge this is the first time it has been applied to a dynamical setting.

\paragraph{Prior work.}
Hamiltonian learning has been the focus of many previous works \cite{Seif2021compressed,evans2019scalablebayesianhamiltonianlearning,li2020hamiltonian,che2021learning,HaahKothariTang2022optimal,yu2023robust,hangleiter2024robustlylearninghamiltoniandynamics,StilckFrança2024,ZubidaYitzhakiEtAl2021optimal,BaireyAradEtAl2019learning, bairey2020learning,GranadeFerrieWiebeCory2012robust,gu2022practical,wilde2022learnH,KrastanovZhouEtAl2019stochastic,Caro_2024,MobusBluhmCaroEtAl2023dissipation,HolzapfelEtAl2015scalable, HuangTongFangSu2023learning,dutkiewicz2023advantage,MiraniHayden2024learning,NiLiYing2024quantum,LiTongNiGefenYing2023heisenberg,BoixoSomma2008parameter,bakshi2024structure,WangLi2024simulation,odake2023universal}. 
Most early works focused on achieving the Heisenberg-limited scaling, i.e., estimating parameters of the Hamiltonian to precision $\epsilon$ with $\mathcal{O}(1/\epsilon)$ cost \cite{higgins2007entanglement}, for small quantum systems, or scalability for large quantum systems but without achieving the Heisenberg-limited scaling. 
Recently, \cite{HuangTongFangSu2023learning} proposed a method to learn an $n$-qubit many-body Hamiltonian with total evolution time $\mathcal{O}(\log(n)/\epsilon)$, where $\epsilon$ is the $\ell^\infty$-error on the coefficients, when the Hamiltonian is low-intersection, i.e., the number of Pauli terms acting on a qubit is a constant that is independent of the system size, and the terms in the Hamiltonian are known. 
This condition can be seen as a slight relaxation of geometric locality. 
The main technique used in \cite{HuangTongFangSu2023learning} is Hamiltonian reshaping, which inserts random Pauli operators during the experiment to obtain an effective Hamiltonian that is easy to learn and preserves some of the coefficients. 
Later works obtained similar results for other control models \cite{dutkiewicz2023advantage} and in bosonic and fermionic scenarios \cite{LiTongNiGefenYing2023heisenberg,NiLiYing2024quantum,MiraniHayden2024learning}. 
These works all assume that the Hamiltonian terms are known and we only need to learn the coefficients.

The first work to consider achieving the Heisenberg-limited scaling for Hamiltonians with unknown terms is \cite{bakshi2024structure}, the authors of which use a bootstrapping approach (as was also done in \cite{dutkiewicz2023advantage}) to gradually refine the estimation of the Hamiltonian. 
A rough estimate of the Hamiltonian is used in the time evolution to partially cancel out the actual Hamiltonian, so that a finer estimate of the residue can be obtained. 
This approach necessitates implementing multi-qubit operations during the time evolution which are adapted based on previous measurement results.

GHZ states and compressed sensing have been used previously to learn dephasing noise with long-range correlation and sparsity \cite{Seif2021compressed}, but the dissipative nature of the noise makes it impossible to achieve the Heisenberg limit. Compared to the setup in \cite{Seif2021compressed}, we also need to deal with the non-commutativity of the Hamiltonian terms.

Our learning protocol takes several important steps beyond prior works. 
First, it only applies single-qubit operations during time evolution, which is an enormous advantage for near-term implementation of algorithms that may have practical importance like Hamiltonian learning. 
Second, we do not require any locality assumptions beyond there being $M$ terms of bounded strength, each of which is $k$-body. 
Lastly, our algorithm is completely non-adaptive. Compared to \cite{bakshi2024structure} in particular, the method in this work uses less total evolution time and number of experiments, which we will illustrate using several physically relevant examples below.

\paragraph{Applications.}
\textit{Sparse Sachdev-Ye model:} We consider a generalization of the Sachdev-Ye (SY) Hamiltonian for qubit system \cite{sachdev1993gapless}. In particular, we remove interaction between qubits independently to obtain a sparse SY model. The Hamiltonian takes the following form:
\begin{equation}
\label{eq:qubit_SY_Hamiltonian}
    H_{SY} = \sum_{1\leq i<j\leq n} \mu_{ij}(\sigma^x_i\sigma^x_j+\sigma^y_i\sigma^y_j+\sigma^z_i\sigma^z_j),
\end{equation}
where the sum over $i, j$ extends over $n$ sites, 
$\mu_{ij} = \chi_{ij}J_{ij}$.
The exchange constants $J_{ij}$ are mutually uncorrelated and selected i.i.d. standard Gaussian random variables. 
Random variables $\chi_{ij}$ are also i.i.d., with $\chi_{ij} = 1$ with probability $\mathfrak{p}$ and $\chi_{ij}=0$ with probability $1-\mathfrak{p}$. The original SY model is recovered when $\mathfrak{p}=1$. We denote the coefficient vector consisting of all $\mu_{ij}$ as $\boldsymbol{\mu}$. We are primarily interested in the case where $\mathfrak{p}=\Omega(1/n^\alpha)$ for $0\leq \alpha<1$.

We will first analyze how prior results can be applied to learn this Hamiltonian.
The method in Ref.~\cite{HuangTongFangSu2023learning}, based on Hamiltonian reshaping to decouple the dynamics, will have a total evolution time of $\exp(\mathcal{O}(n)/\epsilon)$ because each qubit is acted on by $\mathcal{O}(n)$ terms instead of $\mathcal{O}(1)$. 
Ref.~\cite{HaahKothariTang2022optimal} and Ref.~\cite{StilckFrança2024}, based on the cluster expansion, will have unknown complexity due to the lack of Lieb-Robinson bound. 
Ref.~\cite{caro2022learning} will have a total evolution time $\widetilde{\mathcal{O}}(M^{4/p}n\|H_{SY}\|^3/\epsilon^4)$ for $\ell^p$ error to be below $\epsilon$, where $M$ is the number of terms in the Hamiltonian of sparse SY model as explicitly expressed in \eqref{eq:num_terms_sparse_SY}, and $\|H\|=\mathcal{O}(M)$.
Ref.~\cite{bakshi2024structure} will have a total evolution time $\widetilde{\mathcal{O}}(M^{1+1/p}/\epsilon)$ for $\ell^p$ error to be below $\epsilon$, where $M$ is the number of terms in the Hamiltonian of sparse SY model as explicitly expressed in \eqref{eq:num_terms_sparse_SY}, and $\|H\|=\mathcal{O}(M)$.

We then analyze the cost of this task using the protocol in this work.
The probability that the absolute value of any fixed coefficient exceeding $\xi$ is $e^{-\Omega(\xi^2)}$, and consequently by a union bound, the probability of there existing one term in $H_{SY}$ such that the absolute value of its coefficient is larger than $\xi$ is $1-n^2e^{-\Omega(\xi^2)}$.
Then we can choose $\xi = \mathcal{O}(\sqrt{\log(n/\delta_{SY})})$ to ensure that all coefficients are bounded by $\xi$ with probability at least $1-\delta_{SY}/2$. Therefore we can consider the rescaled Hamiltonian $\xi^{-1}H_{SY}$ whose coefficients are all bounded by $1$ with $1-\delta_{SY}/2$ probability.
By the Chernoff-Hoeffding theorem, we can also see that the number of non-zero terms in the Hamiltonian is $n^2 \mathfrak{p} + \mathcal{O}(n\sqrt{\mathfrak{p}\log(1/\delta_{SY})})$ with probability at least $1-\delta_{SY}/2$. Therefore with probability at least $1-\delta_{SY}$, we only need to deal with a Hamiltonian with
\begin{equation}
\label{eq:num_terms_sparse_SY}
    M = n^2 \mathfrak{p} + \mathcal{O}(n\sqrt{\mathfrak{p}\log(1/\delta_{SY})})
\end{equation}
terms, each of which is bounded by $1$. We need to ensure the $\ell^p$-error in the end is $\epsilon/\xi$ to account for the rescaling. Based on the above, we use our main result Theorem~\ref{thm:ham_learning_upper_bound} to arrive at the following result:
\begin{cor}[The sparse Sachdev-Ye model]
    To learn the Hamiltonian  defined in \eqref{eq:qubit_SY_Hamiltonian} with an estimate $\hat{\boldsymbol{\mu}}$ of all coefficients that have $\ell^p$ error ($1\leq p\leq 2$) at most $\epsilon$ using the proposed algorithm with probability at least $(1-\delta)(1-\delta_{SY})$, the total evolution time complexity used is
    \[
        T = \widetilde{\mathcal{O}}\left(\frac{(n^2 \mathfrak{p})^{1/p+1/2}}{\epsilon}\right).
    \]
\end{cor}

Note that in the above we hide the $\mathrm{polylog}(1/\delta_{SY})$ factor.
Compared to previous works \cite{HaahKothariTang2022optimal,HuangTongFangSu2023learning, StilckFrança2024, caro2022learning, bakshi2024structure}, our method has the advantage of guaranteed $\mathrm{poly}(n)$ scaling and achieves the Heisenberg limit. Moreover, it can utilize sparsity to gain further advantage since the cost decreases with $\mathfrak{p}$.

\vspace{0.5em}
\noindent \textit{Hamiltonians with power law interactions:} 
For the power law interaction, the Hamiltonian (on a $D$-dimensional lattice) is
\begin{equation}
    \label{eq:power_law_ham}
    H = \sum_{i\leq j} H_{ij},
\end{equation}
where $H_{ij}$ is supported on qubits $i$ and $j$, and $\|H_{ij}\|\leq 1/(1+d(i,j)^\alpha)$, where $d(i,j)$ is the Euclidean distance between sites $i$ and $j$. We will consider two separate cases: $\alpha\leq D$ or $\alpha> D$. 

For $\alpha\leq D$, which includes Hamiltonians of some interesting physical systems such as the Coulomb potential system ($\alpha=1$, $D=3$), Ref.~\cite{HuangTongFangSu2023learning} will have a total evolution time of $\exp(\mathcal{O}(n)/\epsilon)$ because each qubit is acted on by $\mathcal{O}(n)$ terms instead of $\mathcal{O}(1)$. Ref.~\cite{HaahKothariTang2022optimal} and Ref.~\cite{StilckFrança2024} will have an unknown complexity due to the lack of Lieb-Robinson bound \cite{lieb1972,Tran_2021_Lieb, Hastings_2006, Else_2020, Tran_2019, Chen_2019, Kuwahara_2020, Eldredge_2017, Tran_2020, Tran_2021_Optimal, Nachtergaele_2006_Propagation, Nachtergaele_2006_Lieb, Gong_2014, Storch_2015, Nachtergaele_2008, Pr_mont_Schwarz_2010, PhysRevA.81.062107}. 
Ref.~\cite{caro2022learning} will have $\widetilde{\mathcal{O}}(n^{(1+8/p)}\|H\|^3/\epsilon^4)$ for $\ell^p$ error to be below $\epsilon$, where $\|H\|=\mathcal{O}(n^{2-\alpha/D})$.
Ref.~\cite{bakshi2024structure} will have a total evolution time $\widetilde{\mathcal{O}}(n^{(2/p+2)}/\epsilon$ for $\ell^p$ error to be below $\epsilon$.

For $\alpha>D$, where the decaying effect is more significant and the interactions tend to be more local in space compared to the $\alpha\leq D$ case, Ref.~\cite{HuangTongFangSu2023learning} will have a total evolution time of $\exp(\mathcal{O}(n)/\epsilon)$ because each qubit is acted on by $\mathcal{O}(n)$ terms instead of $\mathcal{O}(1)$, the same reason as in the $\alpha\leq D$ case. 
Ref.~\cite{HaahKothariTang2022optimal} and Ref.~\cite{StilckFrança2024} will have a total evolution time of $\widetilde{\mathcal{O}}(n^{4/p}/\epsilon^2)$ for $\ell^p$ error to be below $\epsilon$ when there is a Lieb-Robinson bound \cite{Tran_2021_Lieb} (each coefficient needs to be estimated to precision $\epsilon n^{-2/p}$), and an unknown complexity when there is not. Ref.~\cite{caro2022learning} will have $\widetilde{\mathcal{O}}(n^{(1+8/p)}\|H\|^3/\epsilon^4)$ for $\ell^p$ error to be below $\epsilon$, where $\|H\|=\mathcal{O}(n)$.
Ref.~\cite{bakshi2024structure} will have a total evolution time $\widetilde{\mathcal{O}}(\min(n^{(2/p+2)}/\epsilon,n^{2(1+\kappa)/p}/\epsilon^{(1+\kappa)}))$ for $\ell^p$ error to be below $\epsilon$, with $\kappa=2D/(\alpha+D)$.

However, with the protocol in this work, we can efficiently learn any Hamiltonian of $\alpha$-power law interactions with the Heisenberg scaling regardless of $\alpha$:
\begin{cor}[The power law interaction Hamiltonians]
    To learn a $\alpha$-power law Hamiltonian of $n$-qubit on a $D$-dimensional lattice as defined in \eqref{eq:power_law_ham} with an estimate $\hat{\boldsymbol{\coef}}$ of all coefficients that have $\ell^p$ error ($1\leq p\leq 2$) at most $\epsilon$ using the proposed algorithm with probability at least $1-\delta$, the total evolution time complexity used is
        \[
            T = \widetilde{\mathcal{O}}\left(\frac{n^{2/p+1}}{\epsilon}\right).
        \]
\end{cor}

The total evolution time above directly follows from the fact that there are $\mathcal{O}(n^2)$ non-zero terms in the Hamiltonian.

We summarize the comparison of results in Table~\ref{tab:comparison}.
\begin{table}[ht!]
\label{tab:comparison}
  \centering 
  \begin{threeparttable}
  \begin{tabular}{cccc}
    \toprule
     \multirow{2}{*}{}& \multirow{2}{*}{Sparse SY model} &  Power law interactions& Power law interactions\\
     & & $\alpha\leq D$ & $\alpha > D$ \\
    \midrule
    \multirow{2}{*}{\cite{HaahKothariTang2022optimal}}  & \multirow{2}{*}{unknown\tnote{1}}  & \multirow{2}{*}{unknown\tnote{1}} & $\widetilde{\mathcal{O}}(n^{4/p}/\epsilon^2)$ (with L-R bound)\\ & & & unknown\tnote{1} (without L-R bound)\\
    \midrule
    
    \cite{HuangTongFangSu2023learning}  & $\exp(\mathcal{O}(n)/\epsilon)$  & $\exp(\mathcal{O}(n)/\epsilon)$ & $\exp(\mathcal{O}(n)/\epsilon)$ \\
    \midrule

    \multirow{2}{*}{\cite{StilckFrança2024}}  & \multirow{2}{*}{unknown\tnote{1}}  & \multirow{2}{*}{unknown\tnote{1}} & $\widetilde{\mathcal{O}}(n^{4/p}/\epsilon^2)$ (with L-R bound)\\ & & & unknown\tnote{1} (without L-R bound)\\
    \midrule
    
    \cite{caro2022learning} & $\widetilde{\mathcal{O}}\left(\frac{(n^2\mathfrak{p})^{3+4/p}n}{\epsilon^4}\right)$ & $\widetilde{\mathcal{O}}\left(\frac{n^{7+8/p-3\alpha/D}}{\epsilon^4}\right)$ & $\widetilde{\mathcal{O}}\left(\frac{n^{4+8/p}}{\epsilon^4}\right)$\\
    \midrule
    
    \cite{bakshi2024structure}\tnote{3}  & $\widetilde{\mathcal{O}}\left(\frac{(n^2 \mathfrak{p})^{1/p+1}}{\epsilon}\right)$  & $\widetilde{\mathcal{O}}\left(\frac{n^{2/p+2}}{\epsilon}\right)$ & $\widetilde{\mathcal{O}}\left(\min\left(\frac{n^{2/p+2}}{\epsilon},\frac{n^{2(1+\kappa)/p}}{\epsilon^{(1+\kappa)}}\right)\right)$\tnote{2}\\
    \midrule

    This work   & $\widetilde{\mathcal{O}}\left(\frac{(n^2 \mathfrak{p})^{1/p+1/2}}{\epsilon}\right)$  & $\widetilde{\mathcal{O}}\left(\frac{n^{2/p+1}}{\epsilon}\right)$ & $\widetilde{\mathcal{O}}\left(\frac{n^{2/p+1}}{\epsilon}\right)$\\
    \bottomrule
  \end{tabular}
  \begin{tablenotes}\footnotesize
    \item[1] Lack of Lieb-Robinson bound.
    \item[2] $1+\kappa=1+2D/(\alpha+D)>1$.
    \item[3] Using the method in \cite[Theorem~5.1]{bakshi2024structure}, the total numbers of experiments for these examples are $\widetilde{\mathcal{O}}(n^4\mathfrak{p}^2)$, $\widetilde{\mathcal{O}}(n^4)$,  and $\widetilde{\mathcal{O}}(\min\{n^4,n^{4\kappa/p}/\epsilon^{2\kappa}\})$ respectively, compared to $\widetilde{\mathcal{O}}(n^2\mathfrak{p})$, $\widetilde{\mathcal{O}}(n^2)$, $\widetilde{\mathcal{O}}(n^2)$ using the method in this work.
    \end{tablenotes}
  \caption{Total evolution times of different methods to achieve $\ell^p$ error $\epsilon$.}
  \end{threeparttable}
\end{table}

\section{Overview}
\label{sec:overview}

Below we will provide an overview of the results and techniques in this work.

\paragraph{Background.}
Our results depend on several important tools developed in previous works. We will briefly introduce these tools here.
Hamiltonian reshaping, proposed in \cite{HuangTongFangSu2023learning} is an approach to reshape the Hamiltonian during time evolution into a new Hamiltonian $H_{\mathrm{eff}}$ that is easy to learn and also contains useful information about the original Hamiltonian. 

In this work we will need to use the Hamiltonian reshaping technique to obtain an effective Hamiltonian that is diagonal with respect to a given single-qubit Pauli eigenbasis. 
More precisely, given a vector $\beta=(\beta_1,\beta_2,\cdots,\beta_n)$, where $\beta_i\in\{x,y,z\}$, we want the effective Hamiltonian to only contain Pauli terms that are either $I$ or $\sigma^{\beta_i}$ ($\sigma^{x,y,z}$ refers to the Pauli-X, Y, and Z matrices respectively) on the $i$th qubit. To achieve this, we apply, with an interval of $\tau$, Pauli operators randomly drawn from the set $\mathcal{K}_{\beta}$ which is a qubit-wise abelian subgroup of the Pauli group on $n$ qubits with $|\mathcal{K}_{\beta}|=2^n$. This results in an effective Hamiltonian
\begin{equation}
\label{eq:effective_ham_overview}
    H_{\mathrm{eff}} = \frac{1}{2^n}\sum_{Q\in \mathcal{K}_{\beta}} QHQ = \sum_{a:P_a\in \mathcal{K}_\beta} \coef_a P_a.
\end{equation}
Note that because $\mathcal{K}_\beta$ is abelian, $H_{\mathrm{eff}}$ consists of commuting Pauli terms, making it easier to learn the coefficients, as will be discussed in Section~\ref{sec:learning_a_commuting_ham}. 
For a more detailed discussion of the Hamiltonian reshaping technique see Section~\ref{sec:hamiltonian_reshaping}.

Compressed sensing provides a powerful tool to reconstruct a sparse high-dimensional vector from few measurements. 
This is relevant to the problem we are studying: the Hamiltonian coefficients form an $M$-sparse vector, but this vector is at the same time high-dimensional, as any of the $\mathcal{O}(n^k)$ $k$-local Pauli terms can be present.

Compressed sensing involves solving the $\ell^1$-minimization problem stated in \eqref{eq:exabsopt}. For the coefficient matrix $\mathbf{A}$, we will use what we call a \emph{weight-$k$ Hadamard matrix} (Definition~\ref{defn:weight-$k$ Hadamard matrix}), which is a $2^n\times D$ submatrix of the Walsh-Hadamard matrix of size $2^n\times 2^n$. Here $D=\sum_{l=0}^k\binom{n}{l}$ is the total number of bit-strings with Hamming weight at most $k$.
In Section~\ref{sec:completely_commuting_hamiltonians}, we will see that the transformation from coefficients to eigenvalues in the completely commuting Hamiltonian $H_{\mathrm{eff}}$ is described by this matrix. Therefore by solving the corresponding $\ell^1$ minimization problem we can reconstruct the sparse coefficient vector with bounded error from few eigenvalue estimates (more precisely, differences between eigenvalues). For a more detailed discussion see Section~\ref{sec:compressed_sensing}. In particular we prove that the weight-$k$ Hadamard matrix has the desired restricted isometry property (RIP).

\paragraph{The Hamiltonian learning protocol.}
We focus on learning a $k$-body Hamiltonian with $M$ terms, which takes the form
\begin{equation}
H = \sum_{P\in\mathbb{P}_n} \coef_P P,
\end{equation}
where $\mathbb{P}_n$ is the set of all Pauli matrices, $-1\leq \coef_P\leq 1$, and $\coef_P \neq 0$ only for $M$ Pauli terms $P_a$ with $\mathrm{wt}(P_a)\leq k$. Here we will briefly describe our protocol for learning this Hamiltonian.

First, the Hamiltonian reshaping procedure discussed in Section~\ref{sec:hamiltonian_reshaping} provides an effective Hamiltonian consisting of completely commuting Pauli terms.
Section~\ref{sec:completely_commuting_hamiltonians} contains a detailed discussion of the eigenvalues and eigenvectors of these Hamiltonians, which we will briefly summarize here.
As stated in \eqref{eq:effective_ham_overview}, the effective Hamiltonian takes the form
$H_{\mathrm{eff}} = \sum_{P\in \mathcal{K}_\beta} \coef_P P$,
where  
\begin{equation}
\label{eq:K_beta_overview}
    \mathcal{K}_{\beta} = \left\{\bigotimes_{j=1}^n Q^j:Q^j=I\text{ or }\sigma^{\beta_j}\right\},
\end{equation}
as given in \eqref{eq:random_unitaries_basis}. 
We note that within the set $\mathcal{K}_\beta$, each element can be represented by a bit-string $c\in\{0,1\}^n$. This is because, for the $i$th qubit, the Pauli operator has a component on it that is either the identity or $\sigma^{\beta_i}$, and we choose $c_i=0$ for the former and $c_i=1$ for the latter.
Based on the above, we can rewrite $H_{\mathrm{eff}}$ as $H_{\mathrm{eff}} = \sum_{c\in\{0,1\}^n} \coef_c P_c$, where $P_c = \bigotimes_i (\sigma^{\beta_i})^{c_i}, c_i\in\{0,1\}$ with the coefficient $\coef_c=\coef_P$ if $P_c=P$, and $\coef_c=0$ otherwise. 
We then collect all coefficients $\coef_c$ into a $2^n$-dimensional vector $\boldsymbol{\coef}_\beta$. Note that this vector contains at most $M$ non-zero entries.

Each eigenvector can similarly be uniquely associated with a bit-string $b$, and a Pauli term $P_c$ contributes to the corresponding eigenvalue by an amount of $(-1)^{c\cdot b}$.
From the above observation, we can see that the relation between the eigenvalues and coefficients of $H_{\mathrm{eff}}$ can be written down in a more compact form:
\begin{equation}
    \label{eq:hadamard_transform_eigenvalue_coef_overview}
    \boldsymbol{\eigen} = \mathbf{H}\boldsymbol{\coef}_\beta,
\end{equation}
where $\mathbf{H}$ is the $2^n\times 2^n$ Hadamard matrix defined in Definition~\ref{defn:weight-$k$ Hadamard matrix}. Therefore we can reconstruct the coefficients of the Hamiltonian by $ \boldsymbol{\coef}_\beta =\mathbf{H}^{-1}\boldsymbol{\eigen} =  \frac{1}{2^n} \mathbf{H}\boldsymbol{\eigen}$. This approach has the drawback that we will need to estimate all $2^n$ of the eigenvalues. 
While techniques such as ``bins and peelings'' for Walsh-Hadamard transforms with sparsity in the transform domain \cite{Scheibler_2015,li2015sprightfastrobustframework} could theoretically reduce the computational cost, as used in \cite{harper2021fast,yu2023robust}, they cannot be used to achieve the Heisenberg-limited scaling and typically assume that the sparse support is randomly distributed across the entire vector space, rather than localized to the subspace with weight at most $k$.

We therefore need to consider an alternative approach, which is through compressed sensing. Using \eqref{eq:hadamard_transform_eigenvalue_coef_overview} to provide the constraint in \eqref{eq:exabsopt} and minimizing the $\ell^1$ norm of the coefficients, we will be able to recover the coefficients $\boldsymbol{\mu}_\beta$.
Quantitatively, if $H_{\mathrm{eff}}$ contains $M$ non-zero terms, then the coefficients can be recovered with  $\widetilde{\mathcal{O}}(M)$ eigenvalue estimates. For a more detailed discussion see Section~\ref{sec:learning_a_commuting_ham}.

To complete our description of the learning protocol for $H_{\mathrm{eff}}$, we only need to describe the experimental setup to learn the eigenvalue corresponding to a given eigenvector, which is a product state. We call such an experiment a phase estimation experiment. It consists of the following steps: (1) Prepare the initial state $\frac{1}{\sqrt{2}}(\ket{0}_{\beta}+\ket{b}_{\beta})$ as in \eqref{eq:experiment_initial_state}.
(2) Let the system evolve for time $t$ while applying random Pauli operators from $\mathcal{K}_\beta$ (defined in \eqref{eq:K_beta_overview}) with interval $\tau$.
(3) Measure the observables $X^b_{\beta}$ or $Y^b_{\beta}$ (by measuring every qubit in either $\sigma^x$, $\sigma^y$, or $\sigma^z$ basis) as defined in \eqref{eq:observable_Xb} and \eqref{eq:observable_Yb} to obtain a $\pm 1$ outcome.

\begin{figure}[t]
    \centering
    \includegraphics[width=0.8\linewidth]{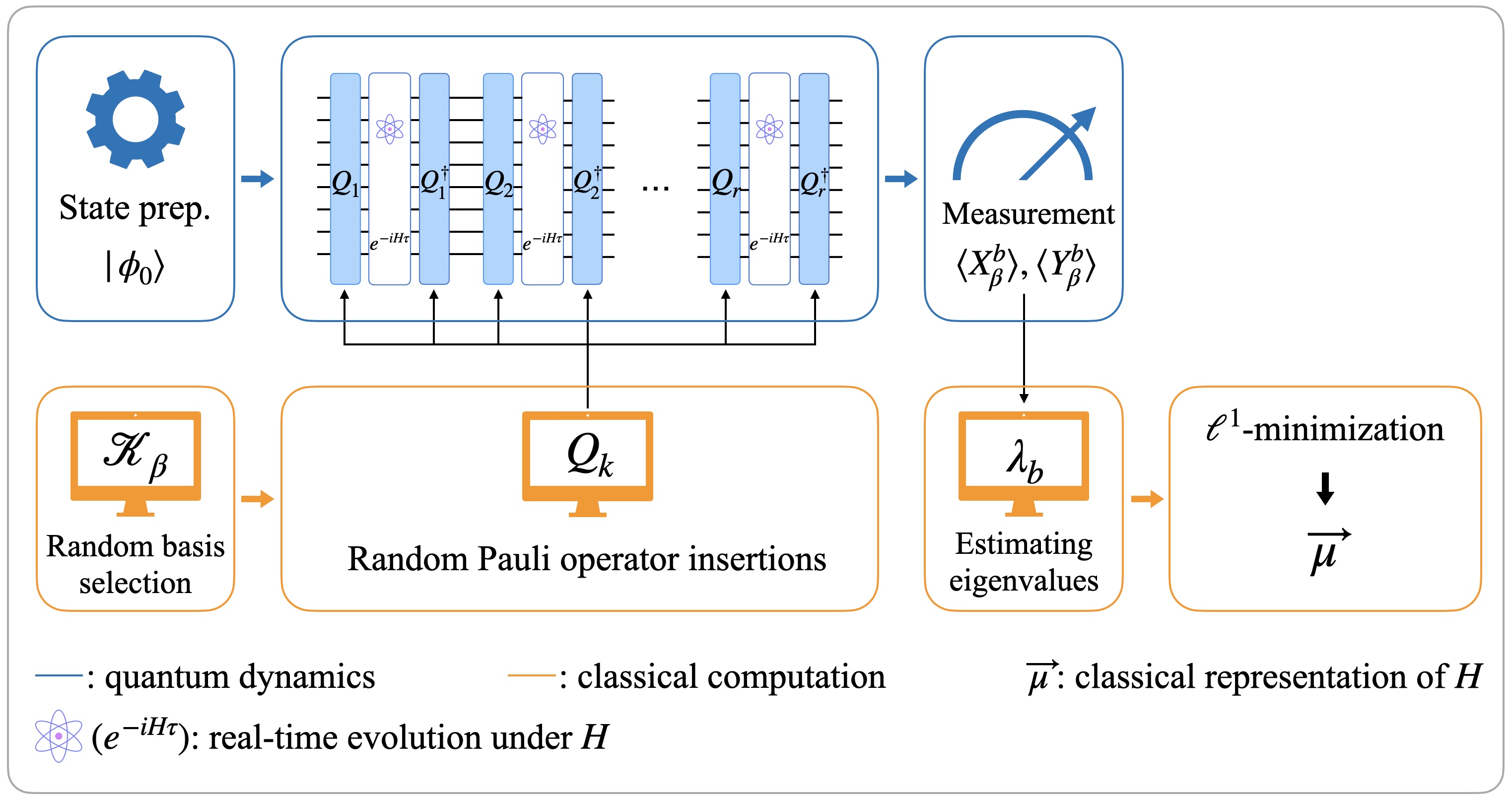}
    \caption{\textbf{The learning protocol.} 
    Quantum processes are illustrated in blue blocks, and classical processes are illustrated in orange blocks. 
    State preparation and measurement (SPAM) errors are taken into account in the experiments.}
    \label{fig:Illustration_of_the_experiment}
\end{figure}

With these experiments we can estimate eigenvalue difference $\eigen_b-\eigen_0$ through the robust frequency estimation protocol introduced in Section~\ref{sec:robust_frequency_estimation}. 
In particular, these estimations have Heisenberg-limited scaling and are SPAM-robust.
We will further discuss how we estimate $\eigen_b-\eigen_0$ and the complexity of the number of independent non-adaptive experiments and total evolution time in Section~\ref{sec:experimental_setup}.
Combined with compressed sensing, we can now learn all coefficients in $H_{\mathrm{eff}}$ efficiently and with Heisenberg-limited scaling. Theorem~\ref{thm:compressed_sensing_main_thm} also tells us that in order to keep the $\ell^1$-error of the coefficient estimates below $\epsilon$ we need to estimate the randomly sampled eigenvalues to precision $\widetilde{\mathcal{O}}(\epsilon M^{-1/2})$, which results in a total evolution time of $\widetilde{\mathcal{O}}(M^{1/2}/\epsilon)$ for eigenvalue estimate using robust phase estimation according to Theorem~\ref{thm:robust_frequency_estimation}. The total evolution time needed for learning $H_{\mathrm{eff}}$ with $\ell^1$-error at most $\epsilon$ is therefore $\widetilde{\mathcal{O}}(M^{1.5}/\epsilon)$.

We can now learn $H_{\mathrm{eff}}$, whose terms are all contained in the set  $\mathcal{K}_\beta$ as defined in \eqref{eq:K_beta_overview}.
In other words, for each Pauli term $P$ in $H_{\mathrm{eff}}$, its component on the $j$th qubit is either $I$ or $\sigma^{\beta_j}$ where $\beta_j\in\{x,y,z\}$.
We will next focus on the way to choose a set of $\beta$ so that all the terms in the Hamiltonians are covered with large probability. We can show that uniformly randomly choosing $\beta$ is a good strategy.
Since each Pauli term $P$ in the Hamiltonian is $k$-body, for a uniformly randomly generated $\beta$, $P\in \mathcal{K}_\beta$ with probability at least $3^{-k}$. 
It therefore takes $L = 3^k\log(M/\delta_{\mathrm{basis}})$ samples of $\beta$ to ensure that all the $M$ terms are included in at least one of the $L$ instances with probability at least $1-\delta_{\mathrm{basis}}$. This ensures that each term $P$ is learned via one of the $H_{\mathrm{eff}}$. 
Taking this overhead $L$ into account, we then arrive at our main result (using the notation of $\mathcal{K}_\beta$ as defined in \eqref{eq:K_beta_overview} and $\mathbb{P}_n^{(k)}$ for the set of Pauli operators on at most $k$ qubits as defined in \eqref{eq:k_body_pauli}):
\begin{thm*}[Learning a $k$-body Hamiltonian containing $M$ terms]
     We assume that the quantum system is evolving under a $k$-body Hamiltonian with $M$ terms (Definition~\ref{defn:k_body_ham}).
    With $N_{\mathrm{exp}}$ independent non-adaptive $(\beta,b_j,t_j,\tau)$-phase estimation experiments (Definition~\ref{defn:phase_estimation_experiment}), $j=1,2,\cdots,N_{\mathrm{exp}}$, with the SPAM error (Definition~\ref{defn:SPAM_err}) satisfying $\epsilon_{\mathrm{SPAM}}\leq 1/(3\sqrt{2})$, we can obtain estimates $\hat{\coef}_P$ for every $P\in \mathbb{P}_n^{(k)}$ such that, with probability at least $1-\delta$
    \begin{equation}
        \left(\sum_{P\in\mathbb{P}_n^{(k)}\setminus\{I\}} |\hat{\coef}_P-\coef_P|^p \right)^{1/p}\leq \epsilon,
    \end{equation}
    for $1\leq p\leq 2$.
    In the above $N_{\mathrm{exp}}$, $\{t_j\}$, and $\tau$ satisfy
    \[
    N_{\mathrm{exp}} = \widetilde{\mathcal{O}}(3^k M),
    \]
    \[
    T = \sum_j t_j =\widetilde{\mathcal{O}}\left(\frac{9^k M^{1/p+1/2}}{\epsilon}\right),
    \]
    and $\tau = \Omega(3^{-k}\epsilon/(M^{1/p+3/2}\log(M/\delta)))$.
\end{thm*}

This theorem is stated as Theorem~\ref{thm:ham_learning_upper_bound} later in the text. 
Although $k = O(1)$ throughout, we keep factors $2^{O(k)}$ explicit to help inform the reader on some $k$ dependence.
More precise results are in Remark~\ref{rem:ham_learning_upper_bound_precise_statement}. 

We also discuss the robustness of our method to modeling errors, i.e., the Hamiltonian may contain more than $M$ terms and may not be exactly $k$-body, in Section~\ref{sec:robustness_to_modeling_errs}. The computational complexity of solving the $\ell^1$-minimization problem in \eqref{eq:exabsopt} is discussed in Section~\ref{sec:Computational_complexity}. Operational interpretations of the $\ell^1$- and $\ell^2$-error metrics are discussed in Section~\ref{sec:operational_interpretation}.

\paragraph{Lower bounds.} 
Besides the learning protocol we also provide a lower bound for the Hamiltonian learning task we consider. 
Specifically we consider the dependence on the tolerated $\ell^1$-error $\epsilon$ and the number of Pauli terms $M$. 
The adaptive experiments are modeled as in \cite{huang2022foundations}, where all experiments form a tree such that the outcome of the experiment at a vertex determines which child leaf to move to, thus determining the next experiment to perform.
Our main tool for the total evolution time lower bound is Assouad's lemma \cite{yu1997assouad}, which provides a lower bound on the achievable $\ell^1$-error given two prerequisites: 
(1) an estimate of how hard it is to distinguish two output probability distributions if they come from two Hamiltonians that differ slightly, and 
(2) a lower bound on the penalty in $\ell^1$-error if such a pair of Hamiltonians are not correctly distinguished. 
The second prerequisite is easy to fulfill, as discussed in the proof of Theorem~\ref{thm:lower_bound}. 
For the first prerequisite, we follow \cite{HuangTongFangSu2023learning} to capture the difficulty of correctly distinguishing the output probability distributions by induction on the tree of adaptive experiments.
We show that for a series of adaptive experiments with Hamiltonians $H$ and $H'$ respectively, if each of them involve a SPAM error of order $\gamma$, then the total variation distance between the output distributions is upper bounded by $1-\gamma^{\|H-H'\|T}$. Combining this with Assouad's lemma allows us to prove the lower bound:
If a learning algorithm can learn the coefficients of an $n$-qubit Hamiltonian $H=\sum_{a=1}^M \coef_aP_a$ with any unknown parameters $\abs{\coef_a}\leq 1$, i.e., after multiple rounds of noisy experiments, the algorithm can estimate any $\boldsymbol{\coef}= \left(\coef_1,\ldots,\coef_M\right)$ to $\epsilon_1$-error in the $\ell^1$-norm in expectation value averaged over experimental outcomes, then, 
    \begin{equation}
        T\geq \frac{M}{\epsilon_1 e \log(1/\gamma)}.
    \end{equation}
Details of the proofs are in Section~\ref{sec:lower_bounds} and Appendix~\ref{sec:Technical lemmas for the lower bound}.

\section{Notations}
\label{sec:notations}
In this work we use $\|\mathbf{x}\|_p$ to denote the $\ell^p$-norm of vector $\mathbf{x}$. 
For a bit string $b$, we use $|b|$ to denote its Hamming weight. 
The Pauli matrices are denoted by $\sigma^{x},\sigma^{y},\sigma^{z}$.
We use the following notation to denote the Pauli eigenstates:
\begin{equation}
    \label{eq:pauli_eigenstates}
    \begin{aligned}
        &\ket{1,z} = \ket{0},\quad \ket{-1,z} = \ket{1},\quad \ket{1,x} = \ket{+},\quad \ket{-1,x} = \ket{-},\\
        &\ket{1,y} = \frac{1}{\sqrt{2}}(\ket{0}+i\ket{1}),\quad \ket{-1,y} = \frac{1}{\sqrt{2}}(\ket{0}-i\ket{1}).
    \end{aligned}
\end{equation}
We denote the set of all $n$-qubit Pauli matrices by $\mathbb{P}_n$:
\begin{equation}
    \label{eq:defn_pauli_matrices}
    \mathbb{P}_n = \left\{\bigotimes_{i=1}^n P_i:P_i=I,\sigma^x,\sigma^y,\text{ or }\sigma^z\right\}.
\end{equation}
For each $P=\bigotimes_{i=1}^n P_i\in \mathbb{P}_n$, we define its weight by
\begin{equation}
    \label{eq:defn_Pauli_weight}
    \mathrm{wt}(P) = |\{i:P_i\neq I\}|.
\end{equation}
We denote the set of Pauli operators that are at most $k$-body by
\begin{equation}
    \label{eq:k_body_pauli}
    \mathbb{P}_n^{(k)} = \{P\in\mathbb{P}_n:\mathrm{wt}(P)\leq k\}.
\end{equation}

\section{Background}
\label{sec:background}

\subsection{Hamiltonian reshaping}
\label{sec:hamiltonian_reshaping}

The basic idea behind \cite{HuangTongFangSu2023learning} is that we can reshape the Hamiltonian during time evolution into a new Hamiltonian that is easy to learn and also contains useful information about the original Hamiltonian. In this section we will provide a brief introduction to this technique.

Hamiltonian reshaping is done by doing random unitary transformations to the time evolution operator corresponding to the unknown Hamiltonian and obtaining the effective Hamiltonian $H_{\mathrm{eff}}$. 
At each time step we apply a unitary transformation that is randomly sampled from an ensemble, which we will describe later in this section. 
As a result the system evolves under a random Hamiltonian transformed under the unitary operation. We assume that these unitary transformations are applied instantaneously.
The same analysis as the one underlying the randomized Hamiltonian simulation algorithm known as qDRIFT \cite{Campbell2019random,BerryChildsEtAl2020time,chen2021concentration} then tells us that the effective Hamiltonian $H_{\mathrm{eff}}$ is the ensemble average of the random Hamiltonians.
This approach is similar in spirit to dynamical decoupling \cite{ViolaLloyd1998dynamical,violaKnillLloyd1999dynamical}, which applies unitary operations in a deterministic manner. However, the randomness in our protocol is helpful as it allows us to efficiently apply random unitaries from a potentially exponentially large set.

In this work we will need to use the Hamiltonian reshaping technique to obtain an effective Hamiltonian that is diagonal with respect to a given single-qubit Pauli eigenbasis. 
More precisely, given a vector $\beta=(\beta_1,\beta_2,\cdots,\beta_n)$, where $\beta_i\in\{x,y,z\}$, we want the effective Hamiltonian to only contain Pauli terms that are either $I$ or $\sigma^{\beta_i}$ on the $i$th qubit. To achieve this, we apply, with an interval of $\tau$, Pauli operators randomly drawn from the set

\begin{equation}
\label{eq:random_unitaries_basis}
    \mathcal{K}_{\beta} = \left\{\bigotimes_{j=1}^n Q^j:Q^j=I\text{ or }\sigma^{\beta_j}\right\}.
\end{equation}
Here we can see that $|\mathcal{K}_{\beta}|=2^n$. $\mathcal{K}_{\beta}$ here is an abelian subgroup of the Pauli group on $n$ qubits. 
Also, $\tau$ needs to be sufficiently small as will be analyzed in Theorem~\ref{thm:hamiltonian_reshaping}.
More precisely, the evolution of the quantum system is described by 
\begin{equation}
    \label{eq:hamiltonian_reshaping}
    Q_r e^{-iH\tau}Q_r\cdots Q_2 e^{-iH\tau}Q_2Q_1 e^{-iH\tau}Q_1,
\end{equation}
where $Q_k$, $k=1,2,\cdots,r$, is uniformly randomly drawn from the set $\mathcal{K}_{\beta}$. In one time step of length $\tau$, the quantum state evolves as
\begin{equation}
\begin{aligned}
    \rho &\mapsto \frac{1}{2^n}\sum_{Q\in \mathcal{K}_\beta} Qe^{-iH\tau}Q\rho Qe^{iH\tau}Q 
    = \rho-\frac{1}{2^n}\sum_{Q\in \mathcal{K}_\beta} i\tau [QHQ,\rho]+\mathcal{O}(\tau^2) \\
    &= \rho - i\tau [H_{\mathrm{eff}},\rho]+\mathcal{O}(\tau^2) 
    =e^{-iH_{\mathrm{eff}}\tau}\rho e^{iH_{\mathrm{eff}}\tau} + \mathcal{O}(\tau^2),
\end{aligned}
\end{equation}
where $H_{\mathrm{eff}}$, the effective Hamiltonian, is
\begin{equation}
    H_{\mathrm{eff}} = \frac{1}{2^n}\sum_{Q\in \mathcal{K}_{\beta}} QHQ.
\end{equation}
In Appendix~\ref{sec:ham_reshape_err_bound} we will provide a bound on how far the actual dynamics deviate from the dynamics induced by the effective Hamiltonian. 

The above can be seen as a linear transformation of the Hamiltonian $H$. It is therefore reasonable to consider what is the effect of the above transformation on each Pauli term in the Hamiltonian. For a Pauli term $P_a$, we note that there are two possible outcomes:

\begin{lem}
    \label{lem:Hamiltonian_reshaping_Pauli_term}
    Let $P$ be a Pauli operator and let $\mathcal{K}_\beta$ be as defined in \eqref{eq:random_unitaries_basis}. Then
    \begin{equation}
    \frac{1}{2^n}\sum_{Q\in \mathcal{K}_\beta} QPQ=
    \begin{cases}
        P,\text{ if } P\in C(\mathcal{K}_\beta),\\
        0,\text{ if } P\notin C(\mathcal{K}_\beta),
    \end{cases}
\end{equation}
where $C(\mathcal{K}_\beta)$ denotes the centralizer of the subgroup $\mathcal{K}_\beta$.
\end{lem}

\begin{proof}
    If $P\in C(\mathcal{K}_\beta)$, then $QPQ=P$ for all $Q\in \mathcal{K}_\beta$, and averaging over all $Q$ therefore yields $P$. 
    If $P\notin \mathcal{K}_\beta$, then there exists $Q_0\in \mathcal{K}_\beta$ such that $PQ_0 = -Q_0 P$ (because two Pauli matrices either commute or anti-commute). 
    Consider the mapping $\phi: \mathcal{K}_\beta\to \mathcal{K}_\beta$ defined by $\phi(Q)=Q_0 Q$. 
    This is a bijection from $\mathcal{K}_\beta$ to itself, and it can be readily checked that if $Q$ commutes with $P$, then $\phi(Q)$ anti-commutes with $P$; if $Q$ anti-commutes with $P$, then $\phi(Q)$ commutes with $P$. 
    Consequently $QPQ=P$ for half of all $Q\in \mathcal{K}_\beta$ and $QPQ=-P$ for the other half. 
    Thus taking the average yields $0$. 
\end{proof}

The above lemma allows us to select which Pauli terms we want to preserve in the Hamiltonian and which ones we want to discard. 
Because $C(\mathcal{K}_\beta)=\mathcal{K}_\beta$, we have
\begin{equation}
    \label{eq:effective_hamiltonian_terms}
    H_{\mathrm{eff}} = \sum_{a:P_a\in \mathcal{K}_\beta} \coef_a P_a.
\end{equation}
Note that because $\mathcal{K}_\beta$ is abelian, $H_{\mathrm{eff}}$ consists of commuting Pauli terms, making it easier to learn the coefficients, as will be discussed in Section~\ref{sec:learning_a_commuting_ham}. 
Also the coefficients $\coef_a$ of $P_a$ that are in $\mathcal{K}_\beta$ are preserved in this effective Hamiltonian. We therefore need to choose a set of $\beta$ so that each Pauli term of $H$ is contained in at least one of the $H_{\mathrm{eff}}$ corresponding to a $\beta$. We will discuss the way to do this in Section~\ref{sec:randomized_basis_selection}.

While Lemma~\ref{lem:Hamiltonian_reshaping_Pauli_term} concerns the uniform average over all $2^n$ elements of a set of commuting Pauli operators, in our learning protocol we will randomly sample from this set. To assess the protocol's accuracy, we will use the following theorem.

\begin{thm}
    \label{thm:hamiltonian_reshaping}
    Let $\beta$ be a length-$n$ string of $x,y,z$, and let $\mathcal{K}_{\beta}$ be as defined in \eqref{eq:random_unitaries_basis}.
    Let $U$ be the random unitary defined in \eqref{eq:hamiltonian_reshaping}. Let $V=e^{-iH_{\mathrm{eff}}t}$ for $H_{\mathrm{eff}}$ given in \eqref{eq:effective_hamiltonian_terms}, and $t=r\tau$.
    We define the quantum channels $\mathcal{U}$ and $\mathcal{V}$ be
    \[
    \mathcal{U}(\rho) = \mathbb{E}[U\rho U^\dag],\quad \mathcal{V}(\rho) = V\rho V^\dag.
    \]
    Then 
    \[
    \|\mathcal{U}-\mathcal{V}\|_{\diamond}\leq \frac{4M^2 t^2}{r}.
    \]
\end{thm}
Here the norm is the diamond norm (or completely-bounded norm), given in terms of the Schatten $L_1$ norm by $\|\mathcal{U}\|_\diamond := \max_{X : \|X\|_1 \le 1}\bigl\|\bigl(\mathcal{U} \otimes \mathbbm{1}_n\bigl)(X)\bigr\|_1$. 
The proof can be found in Appendix~\ref{sec:ham_reshape_err_bound}, where we analyze how the error in each segment of time-evolution contributes to the total diamond distance between the quantum channels $\mathcal{U}$ and $\mathcal{V}$. 
The error in each segment is analyzed using a Taylor expansion.

\subsection{Compressed sensing}
\label{sec:compressed_sensing}
The task of learning a sparse coefficient vector is studied in the field of compressed sensing, which is related to learning the $M$-sparse coefficient vector in the Hamiltonian learning scenario.
Compressed sensing provides a powerful tool to reconstruct a sparse high-dimensional vector from few measurements. 
This is relevant to the problem we are studying: the Hamiltonian coefficients form an $M$-sparse vector, but this vector is at the same time high-dimensional, as any of the $\mathcal{O}(n^k)$ $k$-local Pauli terms can be present.

The basic setup of compressed sensing is as follows: we want to reconstruct a vector $\mathbf{x}\in\CC^{D}$ based on observations $\mathbf{y}\in\CC^{\mathrm{\Gamma}}$ of the form
\begin{equation}
    \mathbf{y} = \mathbf{A}\mathbf{x} + \mathbf{e},
\end{equation}
where $\mathbf{A}$ is a $\Gamma\times D$ matrix and $\mathbf{e}\in \CC^\Gamma$ represents (typically small) noise in the observation. 
When $\Gamma<D$, we will not be able to uniquely determine $\mathbf{x}$ in general. 
However, below we will use $\mathbf{A}$ that has the \emph{restricted isometry property} (RIP), which will be discussed in detail in Appendix~\ref{sec:techinical_lemmas_compressed_sensing}. 
The restricted isometry property of $\mathbf{A}$ guarantees that the $\ell^2$-norm of any sparse vector $\mathbf{x}$ will be approximately preserved after applying $\mathbf{A}$. 
This intuitively implies that we can have enough information to reconstruct $\mathbf{x}$ given access to $\mathbf{Ax}$.
A matrix satisfying the RIP is nearly an isometry when acting on sufficiently sparse vectors, and thus the vector we want to reconstruct is not distorted beyond recognition when transformed through this matrix. 

We further assume that $\mathbf{x}$ is $M$-sparse and solve the following $\ell^1$-minimization problem,
\begin{equation}
\label{eq:l1_minimization}
    \operatornamewithlimits{minimize}_{\mathbf{z} \in \mathbb{C}^D}\|\mathbf{z}\|_1 \quad \text { subject to }\|\mathbf{A z}-\mathbf{y}\|_2 \leq \eta \sqrt{\Gamma},
\end{equation}
with $\Gamma$ being sufficiently large but potentially small compared to $D$. 
The solution $\mathbf{x}^{\sharp}\in \CC^D$ is guaranteed to be close to the actual vector $\mathbf{x}$ if each entry of the error vector $\mathbf{e}$ is at most $\eta$ in absolute value. 
The optimization problem \eqref{eq:l1_minimization} can be solved efficiently as we show in Section~\ref{sec:Computational_complexity}.

In this work, we will focus on what we call a \emph{weight-$k$ Hadamard matrix}, which we define as follows: 
\begin{defn}[Weight-$k$ Hadamard matrix $\mathbf{H}^{(k)}$]
    \label{defn:weight-$k$ Hadamard matrix} 
    Let $\mathbf{H}$ be the Hadamard matrix of size $2^n\times 2^n$, i.e.,
    \[
    \mathbf{H} = \begin{pmatrix}
        1 & 1 \\
        1 & -1
    \end{pmatrix}^{\otimes n}.
    \]
    For each $i=1,2,\cdots,2^n$, we let $b_i$ be the $n$-bit string representing the integer $i-1$. 
    Then the weight-$k$ Hadamard matrix, denoted as $\mathbf{H}^{(k)}$, is the sub-matrix of $\mathbf{H}$ consisting of every $i$th column of $\mathbf{H}$ where $|b_i|\leq k$. 
    Here $|b_i|$ is the Hamming weight of the bit string $b_i$. 
\end{defn}
One can readily see that $\mathbf{H}^{(k)}$ is of size $2^n\times D$ where $D=\sum_{l=0}^k\binom{n}{l}$.
Below we will use this matrix for our compressed sensing task. 
Before we state the result we will first introduce a notation: we define for a vector $\mathbf{x}$ and integers $s,p\geq 0$,
\begin{equation}
    \label{eq:best_s_sparse_approx_err}
    \sigma_s(\mathbf{x})_p = \min_{\mathbf{y}:\|\mathbf{y}\|_0\leq s} \|\mathbf{x}-\mathbf{y}\|_p.
\end{equation}
In other words $\sigma_s(\mathbf{x})_p$ is the smallest error in the $p$-norm one can achieve when approximating $\mathbf{x}$ with an $s$-sparse vector.

\begin{thm}[Compressed sensing with the weight-$k$ Hadamard matrix]
\label{thm:compressed_sensing_main_thm}
    Let $D=\sum_{l=0}^k\binom{n}{l}$.
    Let $\mathbf{A} \in \mathbb{C}^{\Gamma \times D}$ be a matrix whose rows are independently randomly sampled from the rows of $\mathbf{H}^{(k)}$ with replacement.
    Let $\delta_{\mathrm{CS}}\in(0,1)$ and let $M>0$ be an integer. If
    \begin{equation}
    \label{eq: number of samples}
        \Gamma \geq CM\max\{\ln^2(M)\ln(M\ln(D))\ln(D),\ln(1/\delta_{\mathrm{CS}})\},
    \end{equation}
    then, with probability at least $1-\delta_{\mathrm{CS}}$, the following statement holds for every $\mathbf{x} \in \mathbb{C}^D$. Let noisy samples $\mathbf{y}=\mathbf{A x}+\mathbf{e}$ be given with $\|\mathbf{e}\|_2\leq \eta \sqrt{\Gamma}$,
        and let $\mathbf{x}^{\sharp}$ be the solution of the $\ell^1$-minimization problem \eqref{eq:l1_minimization}.
        Then
        \begin{equation}
        \label{eq: reconstruction error}
             \left\|\mathbf{x}-\mathbf{x}^{\sharp}\right\|_p \leq \frac{C_1}{M^{1-1/p}}\sigma_M(\mathbf{x})_1+C_2M^{1/p-1/2} \eta, \quad 1\leq p \leq 2.
        \end{equation}
        $\sigma_M(\mathbf{x})_1$ is defined as in \eqref{eq:best_s_sparse_approx_err}.
        All constants $C, C_1, C_2>0$ are universal. 
\end{thm}
The proof of this theorem is provided in Appendix~\ref{sec:techinical_lemmas_compressed_sensing}.
From the above theorem we can see that, if $\mathbf{x}$ is itself $M$-sparse, then $\sigma_M(\mathbf{x})_1=0$, and the $\ell^1$-minimization allows us to estimate $\mathbf{x}$ to within $C_2M^{1/p-1/2} \eta$ error in the $\ell^p$-norm.

\subsection{Robust frequency estimation}
\label{sec:robust_frequency_estimation}

Taking advantage of the Hamiltonian reshaping procedure, we will be able to identify eigenvectors of the effective Hamiltonian, and as a part of the protocol we want to estimate the corresponding eigenvalue. 
The estimation needs to be performed in a noise-robust way in order to achieve Heisenberg-limited scaling and SPAM robustness.
The robust frequency estimation protocol that we introduce in this section helps us achieve this goal. 
This protocol follows the same idea as the robust phase estimation protocol \cite{KimmelLowYoder2015robust}, but is modified to obtain a good confidence interval rather than minimizing the mean-squared error.

From the phase estimation experiments to be introduced in Section~\ref{sec:experimental_setup}, we will generate signals fluctuating around $\cos(\theta t)$ and $\sin(\theta t)$ where $\theta$ corresponds to the eigenvalue we want to estimate. 
The goal is therefore to reconstruct $\theta$ based on noisy estimates of $\cos(\theta t)$ and $\sin(\theta t)$. 
Through the robust frequency estimation protocol, we can accomplish this task with the cost stated in the following theorem:

\begin{thm}[Robust frequency estimation]
    \label{thm:robust_frequency_estimation}
    Let $\theta\in[-A,A]$.
    Let $X(t)$ and $Y(t)$ be independent random variables satisfying
    \begin{equation}
        \begin{aligned}
            &|X(t)-\cos(\theta t)|< 1/\sqrt{2}, \text{ with probability at least }2/3, \\
            &|Y(t)-\sin(\theta t)|< 1/\sqrt{2}, \text{ with probability at least }2/3.
        \end{aligned}
    \end{equation}
    Then with $K$ independent non-adaptive\footnote{By ``non-adaptive'' we mean that the choice of each $t_j$ does not depend on the value of $X(t_{j'})$ or $Y(t_j')$ for any $j'$.} samples $X(t_1),X(t_2),\cdots,X(t_K)$ and $Y(t_1),Y(t_2),\cdots,Y(t_K)$, $t_j\geq 0$, for
    \begin{equation}
        K=\mathcal{O}(\log(A/\epsilon)(\log(1/q)+\log\log(A/\epsilon))),
    \end{equation}
    \begin{equation}
    \begin{aligned}
        T=\sum_{j=1}^Kt_j=\mathcal{O}((1/\epsilon)(\log(1/q)+\log\log(A/\epsilon))),\quad \max_j t_j=\mathcal{O}(1/\epsilon),
    \end{aligned}
    \end{equation}
    we can obtain a random variable $\hat{\theta}$ such that
    \begin{equation}
        \Pr[|\hat{\theta}-\theta|>\epsilon]\leq q.
    \end{equation}
\end{thm}
For a detailed proof see Appendix~\ref{sec:robust_frequency_estimation_details}.
The main idea is to reduce the estimation problem into a sequence of decision problems with binary outputs in a way similar to bisection. 
Solving each of these decision problems updates our knowledge of $\theta$, and in the process we obtain a smaller and smaller interval containing the exact value. 
We stop the process when the interval is small enough to provide sufficient accuracy.

\section{The Hamiltonian learning protocol}
\label{sec:the_ham_learn_protocol}

The Hamiltonians we will focus on in this work are defined as follows
\begin{defn}[$k$-body Hamiltonians with $M$ terms]
    \label{defn:k_body_ham}
    A $k$-body Hamiltonians with $M$ terms takes the form
    \begin{equation}
    \label{eq:hamiltonian_to_be_learned_all_terms}
    H = \sum_{P\in\mathbb{P}_n} \coef_P P,
    \end{equation}
    where $\mathbb{P}_n$ is the set of all Pauli matrices defined in \eqref{eq:defn_pauli_matrices}, $-1\leq \coef_P\leq 1$, and $\coef_P \neq 0$ only for $M$ Pauli terms $P_a$ with $\mathrm{wt}(P_a)\leq k$.
\end{defn}

\subsection{Completely commuting Hamiltonians}
\label{sec:completely_commuting_hamiltonians}

The Hamiltonian reshaping procedure discussed in Section~\ref{sec:hamiltonian_reshaping} provides an effective Hamiltonian consisting of completely commuting Pauli terms, by which we mean that the Pauli terms commute with each other when acting on any one of the qubits. 
One can readily verify that the Pauli terms in $\mathcal{K}_\beta$ \eqref{eq:random_unitaries_basis} satisfy this requirement. We call these Hamiltonians \emph{completely commuting Hamiltonians}.
Equivalently, we can assign an $\sigma^x$, $\sigma^y$, or $\sigma^z$ basis to each qubit, thus forming a basis for the entire Hilbert space, and then require the Hamiltonian to be diagonal in this basis.

The effective Hamiltonian from Hamiltonian reshaping takes the form
\begin{equation}
    \label{eq:effective_hamiltonian_K_beta}
H_{\mathrm{eff}} = \sum_{P\in \mathcal{K}_\beta} \coef_P P,
\end{equation}
as shown in \eqref{eq:effective_hamiltonian_terms}. The set $\mathcal{K}_\beta$ of completely commuting Pauli terms is specified by the basis given through $\beta=(\beta_1,\beta_2,\cdots,\beta_n)$, as defined in \eqref{eq:random_unitaries_basis}. 
Based on the above, we can rewrite $H_{\mathrm{eff}}$ in another way:
\begin{equation}
    \label{eq:effective_hamiltonian_all_terms}
    H_{\mathrm{eff}} = \sum_{c\in\{0,1\}^n} \coef_c P_c,
\end{equation}
where
\begin{equation}
    P_c = \bigotimes_i (\sigma^{\beta_i})^{c_i},\quad c_i\in\{0,1\},
\end{equation}
the coefficient $\coef_c=\coef_P$ if $P_c=P$, and $\coef_c=0$ otherwise. 
We then collect all coefficients $\coef_c$ into a $2^n$-dimensional vector $\boldsymbol{\coef}_\beta$. Note that this vector contains at most $M$ non-zero entries.

The effective Hamiltonian $H_{\mathrm{eff}}$ has the nice property that its eigenvalues and eigenstates can be written down explicitly. For an $n$-bit string $b=(b_0,b_1,\cdots,b_n)$ we define
\begin{equation}
    \label{eq:bit_string_pauli_basis}
    \ket{b}_{\beta} = \bigotimes_{i=1}^n\ket{(-1)^{b_i},\beta_i}.
\end{equation}
The Pauli eigenstates $\ket{\pm 1,\beta}$ are defined as in \eqref{eq:pauli_eigenstates}.
Then the above $\ket{b}_{\beta}$ is an eigenstate of $H_{\mathrm{eff}}$. We can write down its corresponding eigenvalue, which is 
\begin{equation}
    \label{eq:eigenvalue_effective_hamiltonian}
    \eigen_b = \sum_{c\in\{0,1\}^n} \coef_c (-1)^{c\cdot b},
\end{equation}
where $c\cdot b=c_1b_1+ c_2b_2+\cdots+c_n b_n \pmod{2}$. 
Collecting all eigenvalues $\eigen_b$ into a $2^n$-dimensional vector $\boldsymbol{\eigen}$, we can then write down the relation between the eigenvalues and coefficients in a more compact form:
\begin{equation}
    \label{eq:hadamard_transform_eigenvalue_coef}
    \boldsymbol{\eigen} = \mathbf{H}\boldsymbol{\coef}_\beta,
\end{equation}
where $\mathbf{H}$ is the $2^n\times 2^n$ Hadamard matrix defined in Definition~\ref{defn:weight-$k$ Hadamard matrix}. Therefore we can reconstruct the coefficients of the Hamiltonian by
\begin{equation}
\label{eq:inverse_hadamard_transform}
    \boldsymbol{\coef}_\beta =\mathbf{H}^{-1}\boldsymbol{\eigen} =  \frac{1}{2^n} \mathbf{H}\boldsymbol{\eigen}.
\end{equation}
This approach has the drawback that we will need to estimate all $2^n$ of the eigenvalues. 
While techniques such as "bins and peelings" for Walsh-Hadamard transforms with sparsity in the transform domain \cite{Scheibler_2015,li2015sprightfastrobustframework} could theoretically reduce the computational cost, as used in \cite{harper2021fast,yu2023robust}, they cannot be used to achieve the Heisenberg-limited scaling and typically assume that the sparse support is randomly distributed across the entire vector space, rather than localized to the subspace with weight at most $k$.
We therefore need to consider an alternative approach in Section~\ref{sec:learning_a_commuting_ham}.

\subsection{The experimental setup}
\label{sec:experimental_setup}

As illustrated in Fig~\ref{fig:Illustration_of_the_experiment}, in an experiment we will start from an initial state 
\begin{equation}
\label{eq:experiment_initial_state}
    \ket{\phi_0} =\frac{1}{\sqrt{2}}(\ket{0}_{\beta}+\ket{b}_{\beta})= \frac{1}{\sqrt{2}}\left(\bigotimes_{i=1}^n\ket{1,\beta_i}+\bigotimes_{i=1}^n\ket{(-1)^{b_i},\beta_i}\right),
\end{equation}
where $\ket{b}_\beta$ is defined in \eqref{eq:bit_string_pauli_basis}, and the bit string $b\neq 0$.
This state can be prepared from a $|b|$-qubit GHZ state by applying $\sigma^x$ on some of the qubits. 
The GHZ state can be prepared in constant depth using measurements and feedback \cite{briegel2001persistent,LuLessaKimHsieh2022measurement}.
Because $b$ is uniformly sampled from all bit-strings of length $n$, the average size of the GHZ state that needs to be prepared is $n/2$.
We then let the system evolve while applying random Pauli operators with interval $\tau$ for time $t$ to perform Hamiltonian reshaping as discussed in \eqref{eq:hamiltonian_reshaping}. 
Because $\ket{0}_{\beta}$ and $\ket{b}_{\beta}$ are both eigenstates of the effective Hamiltonian $H_{\mathrm{eff}}$ corresponding to eigenvalues $\eigen_0$ and $\eigen_b$ (defined in \eqref{eq:eigenvalue_effective_hamiltonian}), as discussed in  Section~\ref{sec:completely_commuting_hamiltonians}, at time $t$ we will approximately obtain the state
\begin{equation}
\label{eq:exact_state_phase_estimation_experiment}
    \ket{\phi_t} = \frac{1}{\sqrt{2}}(e^{-i\eigen_0 t}\ket{0}_\beta+e^{-i\eigen_b t}\ket{b}_\beta).
\end{equation}
In the end, we then measure the observable 
\begin{equation}
\label{eq:observable_Xb}
    X^b_\beta = \bigotimes_{i=1}^n Q_i,
\end{equation}
where each $Q_i$ is chosen to be a Pauli operator such that
\begin{equation}
    \label{eq:Q_i_choice_X}
    Q_i\ket{1,\beta_i} = \ket{(-1)^{b_i},\beta_i},\quad Q_i\ket{(-1)^{b_i},\beta_i}=\ket{1,\beta_i}.
\end{equation}
Such $Q_i$ operators can be constructed as follows: in the case of $b_i=1$, if $\beta_i=z$ then $Q_i=\sigma^x$, and if $\beta_i=y$ or $x$ then $Q_i=\sigma^z$; in the case of $b_i=0$ then $Q_i=I$.
This observable therefore satisfies
\[
X^b_\beta \ket{0}_\beta = \ket{b_\beta},\quad X^b_\beta\ket{b_\beta}=\ket{0_\beta}.
\]
The expectation value will be
\begin{equation}
\label{eq:X_expectation_value}
    \bra{\phi_t}X^b_{\beta}\ket{\phi_t} = \cos((\eigen_b-\eigen_0)t).
\end{equation}
Note that in order to measure $X^b_{\beta}$ we only need to measure each individual $Q_i$ to obtain a $\pm 1$ outcome and then multiply them together.

Similarly, we can measure the observable 
\begin{equation}
\label{eq:observable_Yb}
    Y^b_\beta = \bigotimes_{i=1}^n Q_i,
\end{equation} 
where the $Q_i$ operators are chosen slightly differently. For the first $i$ for which $b_i=1$, we choose $Q_i$ so that
\begin{equation}
    \label{eq:Q_i_choice_Y}
    Q_i\ket{1,\beta_i} = i\ket{(-1)^{b_i},\beta_i},\quad Q_i\ket{(-1)^{b_i},\beta_i}=-i\ket{1,\beta_i}.
\end{equation}
Note that we can choose $Q_i=\sigma^y$ if $\beta_i=z$, $Q_i=\sigma^x$ if $\beta_i=y$, and $Q_i=-\sigma^y$ if $\beta_i=x$ to satisfy this requirement.
For all other $i$ we choose $Q_i$ to satisfy \eqref{eq:Q_i_choice_X}. The resulting operator $Y^b_{\beta}$ satisfies
\[
Y^b_\beta \ket{0}_\beta = i\ket{b_\beta},\quad Y^b_\beta\ket{b_\beta}=-i\ket{0_\beta}.
\]
Therefore the expectation value is
\begin{equation}
\label{eq:Y_expectation_value}
    \bra{\phi_t}Y^b_{\beta}\ket{\phi_t} = -\sin((\eigen_b-\eigen_0)t).
\end{equation}

We summarize the above experiment in the following definition
\begin{defn}[Phase estimation experiment]
    \label{defn:phase_estimation_experiment}
    We call the procedure below a $(\beta,b,t,\tau)$-phase estimation experiment:
    \begin{enumerate}
        \item Prepare the initial state $\frac{1}{\sqrt{2}}(\ket{0}_{\beta}+\ket{b}_{\beta})$ as in \eqref{eq:experiment_initial_state}.
        \item Let the system evolve for time $t$ while applying random Pauli operators from $\mathcal{K}_\beta$ (defined in \eqref{eq:random_unitaries_basis}) with interval $\tau$.
        \item Measure the observables $X^b_{\beta}$ or $Y^b_{\beta}$ (by measuring every qubit in either $\sigma^x$, $\sigma^y$, or $\sigma^z$ basis) as defined above to obtain a $\pm 1$ outcome.
    \end{enumerate}
\end{defn}

The goal of the above experiment is to estimate $\eigen_b-\eigen_0$, for which we also need to use the robust frequency estimation protocol introduced in Section~\ref{sec:robust_frequency_estimation}. 
Below we will discuss how this is done, and analyze the effect of the Hamiltonian reshaping error and the state preparation and measurement (SPAM) error. 

We model the SPAM error as follows: 
\begin{defn}
    \label{defn:SPAM_err}
    The preparation of the initial state involves an error channel $\mathcal{E}_{\mathrm{prep}}$ applied after the ideal state preparation channel, and the measurement involves an error channel $\mathcal{E}_{\mathrm{meas}}$  applied before the ideal measurement channel. 
    We assume that
    \[
    \|\mathcal{E}_{\mathrm{prep}}-\mathcal{I}\|_{\diamond}+\|\mathcal{E}_{\mathrm{meas}}-\mathcal{I}\|_{\diamond}\leq \epsilon_{\mathrm{SPAM}}.
    \]
\end{defn}

Through a $(\beta,b,t,\tau)$-phase estimation experiment, if we measure $X^b_{\beta}$ in the end, we will obtain a random variable $s^x(t)\in\{\pm 1\}$. If we had the exact state $\ket{\phi_t}$ as defined in \eqref{eq:exact_state_phase_estimation_experiment}, then we would have $\mathbb{E}[s^x(t)] = \cos((\eigen_b-\eigen_0)t)$. However, due to the Hamiltonian reshaping error and the SPAM error, we have
\begin{equation}
\label{eq:phase_estimation_experiment_expectation_deviation}
    |\mathbb{E}[s^x(t)]-\cos((\eigen_b-\eigen_0)t)|\leq \frac{4M^2 t^2}{r} + \epsilon_{\mathrm{SPAM}}.
\end{equation}

 Notice that the first term on the right-hand side comes from the Hamiltonian reshaping error, where we set $t=r\tau$ as in Theorem~\ref{thm:hamiltonian_reshaping}. Moreover the variance of $s^x(t)$ is at most $1$ because it can only take values $\pm 1$.
We assume that $\epsilon_{\mathrm{SPAM}}\leq 1/(3\sqrt{2})$, and choose $r$ to be $r= \mathcal{O}(M^2 t^2)$ so that 
\[
\frac{4M^2 t^2}{r} < \frac{1}{3\sqrt{2}}.
\]
Note that $\tau$ and $r$ are related through $\tau=t/r$.
Then we have
\[
|\mathbb{E}[s^x(t)]-\cos((\eigen_b-\eigen_0)t)| < \frac{2}{3\sqrt{2}}.
\]
We then take $54$ independent samples of $s^x(t)$ and average them, denoting the sample average by $X(t)$. 
By Chebyshev's inequality, this ensures
\[
\Pr[|X(t)-\mathbb{E}[s^x(t)]|\geq 1/(3\sqrt{2})]=\Pr[|X(t)-\mathbb{E}[X(t)]|\geq 1/(3\sqrt{2})]\leq \frac{1}{54\times 1/(3\sqrt{2})^2}=\frac{1}{3}.
\]
Therefore combining the above with \eqref{eq:phase_estimation_experiment_expectation_deviation} we have
\[
\Pr[|X(t)-\cos((\eigen_t-\eigen_0)t)|\geq 1/\sqrt{2}]\leq 1/3.
\]
This guarantees estimating $\cos((\eigen_t-\eigen_0))t$ to a constant $1/\sqrt{2}$ accuracy with at least $2/3$ probability, which gives us the $X(t)$ required in the robust frequency estimation protocol in Theorem~\ref{thm:robust_frequency_estimation}. The $Y(t)$ in Theorem~\ref{thm:robust_frequency_estimation} can be similarly obtained. 
We also note that because $|\eigen_b-\eigen_0|\leq 2M$, we can set $A=2M$ in Theorem~\ref{thm:robust_frequency_estimation}. Therefore we can state the following for the phase estimation experiment (using the notation of $\mathcal{K}_\beta$ as defined in \eqref{eq:random_unitaries_basis} and $\mathbb{P}_n^{(k)}$ for the set of Pauli operators on at most $k$ qubits as defined in \eqref{eq:k_body_pauli}):

\begin{thm}
    \label{thm:phase_estimation_experiment}
    We assume that the quantum system is evolving under a $k$-body Hamiltonian with $M$ terms (Definition~\ref{defn:k_body_ham}).
    With $N_{\mathrm{exp}}$ independent non-adaptive $(\beta,b_j,t_j,\tau_j)$-phase estimation experiments (Definition~\ref{defn:phase_estimation_experiment}), $j=1,2,\cdots,N_{\mathrm{exp}}$, with the SPAM error (Definition~\ref{defn:SPAM_err}) satisfying $\epsilon_{\mathrm{SPAM}}\leq 1/(3\sqrt{2})$, we can obtain an estimate $\hat{\theta}$ such that
    \[
    \Pr[|\hat{\theta}-(\eigen_b-\eigen_0)|\geq \eta] \leq q,
    \]
    where $\eigen_0$ and $\eigen_b$ (defined in \eqref{eq:eigenvalue_effective_hamiltonian}) are eigenvalues of the effective Hamiltonian $H_{\mathrm{eff}}$ defined in \eqref{eq:effective_hamiltonian_all_terms}.
    In the above $N_{\mathrm{exp}}$, $\{t_j\}$, and $\{\tau_j\}$ satisfy
    \[
    N_{\mathrm{exp}} = \mathcal{O}(\log(M/\eta)(\log(1/q)+\log\log(M/\eta))),
    \]
    \begin{equation}
    \begin{aligned}
        T=\sum_{j=1}^{N_{\mathrm{exp}}} t_j=\mathcal{O}\left(\frac{1}{\eta}(\log(1/q)+\log\log(M/\eta))\right),\quad \max_j t_j=\mathcal{O}(1/\eta),
    \end{aligned}
    \end{equation}
    and $\tau_j = \Omega(1/(M^2 t_j))$. $D=\sum_{l=0}^k \binom{n}{l}=\Theta(n^k)$.
\end{thm}

\subsection{Learning a completely commuting Hamiltonian}
\label{sec:learning_a_commuting_ham}

In the last section we have shown how to estimate eigenvalue differences $\eigen_b-\eigen_0$ through phase estimation experiments. 
Up to fixing a global phase, we have the eigenvalues $\eigen_b$ (we will discuss the global phase issue in more detail later in this section). 
These eigenvalues can be used to recover the Hamiltonian coefficients though an inverse Hadamard transform as described in \eqref{eq:inverse_hadamard_transform}, but a naive approach would not be efficient as there are $2^n$ different eigenvalues to estimate.
Note that we have not yet used the information that the coefficient vector $\boldsymbol{\coef}_\beta$ contains at most $M$ non-zero entries, and compressed sensing is a tool that utilizes exactly this information.

Before we apply compressed sensing, we first note that all the nonzero entries of $\boldsymbol{\coef}_\beta$ of the unknown $k$-body Hamiltonian must (by assumption) correspond to $k$-body Pauli terms, which already helps us reduce the size of the linear system. 
Let $\boldsymbol{\coef}_\beta^{(k)}$ be the $D$-dimensional vector consisting of the entries of $\boldsymbol{\coef}$ that correspond to $k$-body Pauli terms, where $D=\sum_{l=0}^k \binom{n}{l}$. 
Then we can write \eqref{eq:hadamard_transform_eigenvalue_coef} as
\[
\boldsymbol{\eigen} = \mathbf{H}^{(k)}\boldsymbol{\coef}_\beta^{(k)},
\]
where $\mathbf{H}^{(k)}$ is the $2^n\times D$ weight-$k$ Hadamard matrix defined in Definition~\ref{defn:weight-$k$ Hadamard matrix}.

As is commonly done in compressed sensing, we generate $\Gamma$ independent samples of the rows of $\mathbf{H}^{(k)}$, use those rows to form a matrix $\mathbf{A}$, and use the corresponding entries of $\boldsymbol{\eigen}$ to form a vector $\mathbf{y}$. Then we solve the $\ell^1$-minimization problem in \eqref{eq:l1_minimization}, where $\eta$ is the estimation error upper bound for each entry of $\mathbf{y}$. When $\Gamma$ satisfies \eqref{eq: number of samples}, then by Theorem~\ref{thm:compressed_sensing_main_thm} the solution of the $\ell^1$-minimization problem \eqref{eq:l1_minimization} is guaranteed to be close to the actual coefficients. More precisely, denoting the solution by $\hat{\boldsymbol{\coef}}_\beta^{(k)}$, we have $\|\hat{\boldsymbol{\coef}}^{(k)}_\beta-\boldsymbol{\coef}_\beta^{(k)}\|_p\leq C_2 M^{1/p-1/2}\eta$ for $1\leq p\leq 2$.

In the above we assumed that we can estimate eigenvalues $\eigen_b$ for any $b\in\{0,1\}^n$ to precision $\eta$. This is impossible given the ambiguity coming from the global phase.  However, we will explain why this does not affect our ability to learn the Hamiltonian. Note that the phase estimation experiment allows us to estimate $\eigen_b-\eigen_0$. We can therefore fix the global phase by first assuming $\eigen_0=0$ to obtain $\eigen_b$. This can be done by shifting the effective Hamiltonian, so that we are in fact estimating the eigenvalues of
\[
\Tilde{H}_{\mathrm{eff}} = \left(-\sum_{c\neq 0}\coef_c\right)I+\sum_{c\neq 0}\coef_c P_c.
\]
Therefore the above procedure guarantees that we can learn all the coefficients $\coef_c$ accurately. Note that these $\coef_c$ correspond to all Pauli terms that are $k$-body and are in $\mathcal{K}_{\beta}$. Therefore we have obtained a set of estimates $\hat{\coef}_P$ for all $P\in\mathcal{K}_{\beta}$ and $\mathrm{wt}(P)\leq k$ that satisfies
\begin{equation}
\label{eq:completely_commuting_ham_accuracy_guarantee}
    \left(\sum_{P\in \mathcal{K}_{\beta}:1\leq \mathrm{wt}(P)\leq k} |\hat{\coef}_P-\coef_P|^p\right)^{1/p}\leq C M^{1/p-1/2}\eta,
\end{equation}
for $1\leq p\leq 2$ and a universal constant $C$.

Multiplying the overhead of compressed sensing as described in \eqref{eq: number of samples} by the number of samples $N_{\mathrm{exp}}$ and total evolution time $T$ described in Theorem~\ref{thm:phase_estimation_experiment}, we arrive at the following theorem (using the notation of $\mathcal{K}_\beta$ as defined in \eqref{eq:random_unitaries_basis} and $\mathbb{P}_n^{(k)}$ for the set of Pauli operators on at most $k$ qubits as defined in \eqref{eq:k_body_pauli}):
\begin{thm}[Learning the coefficients of terms in $\mathcal{K}_\beta$]
    \label{thm:learning_completely_commuting_hams}
    We assume that the quantum system is evolving under a $k$-body Hamiltonian with $M$ terms (Definition~\ref{defn:k_body_ham}).
    With $N_{\mathrm{exp}}$ independent non-adaptive $(\beta,b_j,t_j,\tau_j)$-phase estimation experiments (Definition~\ref{defn:phase_estimation_experiment}), $j=1,2,\cdots,N_{\mathrm{exp}}$, with the SPAM error (Definition~\ref{defn:SPAM_err}) satisfying $\epsilon_{\mathrm{SPAM}}\leq 1/(3\sqrt{2})$, we can obtain estimates $\hat{\coef}_P$ for every $P\in \mathcal{K}_\beta\cap \mathbb{P}_n^{(k)}$, such that \eqref{eq:completely_commuting_ham_accuracy_guarantee} holds with probability at least $1-\delta_{\mathrm{CS}}-\Gamma q$.
    In the above $N_{\mathrm{exp}}$, $\{t_j\}$, and $\{\tau_j\}$ satisfy
    \[
    N_{\mathrm{exp}} = \mathcal{O}(\Gamma\log(M/\eta)(\log(1/q)+\log\log(M/\eta))),
    \]
    \begin{equation}
    \begin{aligned}
        T=\sum_{j=1}^{N_{\mathrm{exp}}} t_j=\mathcal{O}\left(\frac{\Gamma}{\eta}(\log(1/q)+\log\log(M/\eta))\right),\quad \max_j t_j=\mathcal{O}(1/\eta),
    \end{aligned}
    \end{equation}
    \begin{equation}
        \Gamma = \mathcal{O}\left(M(\ln^2(M)\ln(M\ln(D))\ln(D)+\ln(1/\delta_{\mathrm{CS}}))\right),
    \end{equation}
    and $\tau_j = \Omega(1/(M^2 t_j))$. $D=\sum_{l=0}^k \binom{n}{l}=\Theta(n^k)$.
\end{thm}

The success probability in the theorem comes from considering all possible ways the protocol may fail and then taking the union bound: it may fail because the sub-sampled matrix in compressed sensing does not have the RIP property, which contributes a failure probability of $\delta_{\mathrm{CS}}$, and it may fail because any one of the eigenvalue estimations is not accurate enough, which contributes a failure probability of $\Gamma q$. Note that the failure probability can be made small with little overhead, as the dependence on $1/q$ and $1/\delta_{\mathrm{CS}}$ are both logarithmic.

\subsection{Randomized basis selection}
\label{sec:randomized_basis_selection}

In the previous section we have discussed how to learn all the coefficients corresponding to Hamiltonian terms contained in the set $\mathcal{K}_\beta$ (defined in \eqref{eq:random_unitaries_basis}). In this section, we will focus on the way to choose a set of $\mathcal{K}_\beta$ so that all the terms in the Hamiltonians are covered with large probability. We will show that uniformly randomly choosing $\mathcal{K}_\beta$ is a good strategy.

To see this, we observe that by uniformly randomly sampling a $\beta\in\{x,y,z\}^{n}$, for a fixed Pauli operator $P\in\mathbb{P}_n$, 
\[
\Pr[P\in\mathcal{K}_{\beta}] = 3^{-\operatorname{wt}(P)}\geq 3^{-k}.
\]
Therefore, each time we learn a randomly sample $\beta$, there is at least $3^{-k}$ probability that we will learn the coefficient $\coef_P$ for $\mathrm{wt}(P)\leq k$. 
If we generate $L$ independent samples of $\beta$, then the probability of a term $P$ being not covered in any of the $L$ instances is at most $(1-3^{-k})^L$.
By the union bound, to ensure that all the $M$ terms are included in at least one of the $L$ instances with probability at least $1-\delta_{\mathrm{basis}}$, we need
\begin{equation}
    M(1-3^{-k})^L \leq M \exp\bigl(3^{-k}L) \leq \delta_{\mathrm{basis}}.
\end{equation}
It is thus sufficient to choose
$L = 3^k\log(M/\delta_{\mathrm{basis}})$.
Taking this overhead $L$ into account, we then arrive at our main result (using the notation of $\mathcal{K}_\beta$ as defined in \eqref{eq:random_unitaries_basis} and $\mathbb{P}_n^{(k)}$ for the set of Pauli operators on at most $k$ qubits as defined in \eqref{eq:k_body_pauli}):
\begin{thm}[Learning a $k$-body Hamiltonian containing $M$ terms]
    \label{thm:ham_learning_upper_bound}
     We assume that the quantum system is evolving under a $k$-body Hamiltonian with $M$ terms (Definition~\ref{defn:k_body_ham}).
    With $N_{\mathrm{exp}}$ independent non-adaptive $(\beta,b_j,t_j,\tau)$-phase estimation experiments (Definition~\ref{defn:phase_estimation_experiment}), $j=1,2,\cdots,N_{\mathrm{exp}}$, with the SPAM error (Definition~\ref{defn:SPAM_err}) satisfying $\epsilon_{\mathrm{SPAM}}\leq 1/(3\sqrt{2})$, we can obtain estimates $\hat{\coef}_P$ for every $P\in \mathbb{P}_n^{(k)}$ such that, with probability at least $1-\delta$
    \begin{equation}
    \label{eq:p_norm_err_bound}
        \left(\sum_{P\in\mathbb{P}_n^{(k)}\setminus\{I\}} |\hat{\coef}_P-\coef_P|^p \right)^{1/p}\leq \epsilon,
    \end{equation}
    for $1\leq p\leq 2$.
    In the above $N_{\mathrm{exp}}$, $\{t_j\}$, and $\tau$ satisfy
    \[
    N_{\mathrm{exp}} = \widetilde{\mathcal{O}}(3^k M),
    \]
    \[
    T = \sum_j t_j =\widetilde{\mathcal{O}}\left(\frac{9^k M^{1/p+1/2}}{\epsilon}\right),
    \]
    and $\tau = \Omega(3^{-k}\epsilon/(M^{1/p+3/2}\log(M/\delta)))$.
\end{thm}

In the above, we choose $\tau$ to be the same for all experiments for simplicity. 
One can also choose $\tau$ differently for each experiment to reduce the total number of Pauli operators that are needed, as is done in the remark below. 
We are primarily concerned with the cases with $p=1$ and $p=2$, which have specific operational interpretations that we will discuss in Section~\ref{sec:operational_interpretation}, but our result holds for all $1\leq p\leq 2$.

\begin{rem}
    \label{rem:ham_learning_upper_bound_precise_statement}
    In the above we have used the $\tilde{O}$-notation to hide polylogarithmic factors. Here we provide a more precise statement. We constructed a protocol that with probability at least $1-L\delta_{\mathrm{CS}}-L\Gamma q-\delta_{\mathrm{basis}}$, can obtain estimates $\hat{\coef}_P$ for $P\in \mathbb{P}_n^{(k)}$ such that
    \[
    \left(\sum_{P\in\mathbb{P}^{(k)}_n\setminus\{I\}} |\hat{\coef}_P-\coef_P|^p \right)^{1/p}\leq CL M^{1/p-1/2}\eta,
    \]
    for $1\leq p\leq 2$ and a positive universal constant $C$. 
    This protocol uses
    \[
    N_{\mathrm{exp}} = \mathcal{O}(L\Gamma\log(M/\eta)(\log(1/q)+\log\log(M/\eta))),
    \]
    experiments and 
    \begin{equation}
    \begin{aligned}
        T=\sum_{j=1}^{N_{\mathrm{exp}}} t_j=\mathcal{O}\left(\frac{L\Gamma}{\eta}(\log(1/q)+\log\log(M/\eta))\right),\quad \max_j t_j=\mathcal{O}(1/\eta),
    \end{aligned}
    \end{equation}
    total evolution time.
    In the above $\{t_j\}$, $\{\tau_j\}$, $\Gamma$, and $L$ satisfy
    \begin{equation}
        \Gamma = \mathcal{O}\left(M(\ln^2(M)\ln(M\ln(D))\ln(D)+\ln(1/\delta_{\mathrm{CS}}))\right),
    \end{equation}
    \begin{equation}
        L = \mathcal{O}\left(3^k\log(M/\delta_{\mathrm{basis}})\right),
    \end{equation}
    and $\tau_j = \Omega(1/(M^2 t_j))$. $D=\sum_{l=0}^k \binom{n}{l}=\Theta(n^k)$.
\end{rem}

\subsection{Robustness to modeling errors}
\label{sec:robustness_to_modeling_errs}

In Theorem~\ref{thm:compressed_sensing_main_thm} we have shown that the Hamiltonian learning protocol is robust against SPAM error. 
In this section we will discuss the robustness against modeling errors. 
More precisely, we will ask two questions: 
(1) What if the Hamiltonian contains more than $M$ terms? 
(2) What if the terms in the Hamiltonian are not exactly $k$-body? We will show that our learning protocol can still work if these two types of modeling errors occur. 

The robustness against the first type of modeling error, i.e., that the Hamiltonian contains more than $M$ terms, but all terms are still $k$-body, is a direct consequence of Theorem~\ref{thm:compressed_sensing_main_thm}, in which an additional error term $\sigma_M(\mathbf{x})_1$ (defined in \eqref{eq:best_s_sparse_approx_err}) is included to account for the error of approximating a potentially non-sparse vector with an $M$-sparse one. In the context of Theorem~\ref{thm:compressed_sensing_main_thm}, this means an extra term is included in the estimation error for the coefficients, which gives us
\begin{equation}
    \left(\sum_{P\in\mathbb{P}_n^{(k)}\setminus\{I\}} |\hat{\coef}_P-\coef_P|^p \right)^{1/p}\leq \frac{C 3^k\sigma_M(\boldsymbol{\coef})_1\log(M/\delta)}{M^{1-1/p}}+\epsilon,
\end{equation}
for some universal constant $C$, where $\epsilon$ is the error bound without modeling errors as defined in \eqref{eq:p_norm_err_bound}, and $\boldsymbol{\coef}$ is the vector consisting of $\coef_P$ for all $P\in\mathbb{P}_n\setminus\{I\}$. Here we have used the fact that the best $M$-sparse approximation error for terms in $\mathcal{K}_\beta$ must be smaller than or equal to that for the whole coefficient vector. From the above we can see that if the Hamiltonian can be well-approximated by one with $M$ $k$-body terms, then we can still obtain an accurate estimate for its leading coefficients.

For the second type of modeling error, we will first focus on the procedure of learning terms in $\mathcal{K}_\beta$ described in Section~\ref{sec:learning_a_commuting_ham}. The coefficients for all terms in $\mathcal{K}_\beta$ are collected into a vector $\boldsymbol{\coef}_\beta$. Those that are also $k$-body are collected into $\boldsymbol{\coef}_\beta^{(k)}$, while the rest form a vector $\bar{\boldsymbol{\coef}}_\beta^{(k)}$. Here $\boldsymbol{\coef}_\beta$ contains $D=\sum_{l=0}^k\binom{n}{l}$ entries.

We denote by $\widetilde{\mathbf{A}}$ the submatrix of the original Wash-Hadamard matrix $\mathbf{H}$ consisting of the $\Gamma$ rows corresponding to the randomly sampled rows of $A$ in Theorem~\ref{thm:compressed_sensing_main_thm}. We can write $\widetilde{\mathbf{A}}$ as $\widetilde{\mathbf{A}}=[\mathbf{A}~\mathbf{B}]$, where $\mathbf{A}\in\mathbb{R}^{\Gamma\times D}$ is the submatrix corespond to rows of the weight-$k$ Hadamard matrix as defined in Definition~\ref{defn:weight-$k$ Hadamard matrix}, and $\mathbf{B}\in\mathbb{R}^{\Gamma\times (2^n-D)}$. Then the eigenvalue estimates $\mathbf{y}$ and the coefficients are related through
\begin{equation}
    \mathbf{y}  = \widetilde{\mathbf{A}}\boldsymbol{\coef}_\beta + \mathbf{e} = \mathbf{A}\boldsymbol{\coef}_\beta^{(k)} + \mathbf{B}\bar{\boldsymbol{\coef}}_\beta^{(k)} +  \mathbf{e},
\end{equation}
where $\mathbf{e}$ is the vector consisting of errors on each eigenvalue estimate, and using the phase estimation experiments described in Section~\ref{sec:experimental_setup} we can ensure that each entry of $\mathbf{e}$ is at most $\eta$ in absolute value, thus ensuring $\|\mathbf{e}\|\leq \sqrt{\Gamma}\eta$.
We therefore have
\begin{equation}
    \|\mathbf{y}-\mathbf{A}\boldsymbol{\coef}_\beta^{(k)}\|_2\leq \|\mathbf{B}\bar{\boldsymbol{\coef}}_\beta^{(k)}\|_2 + \sqrt{\Gamma}\eta.
\end{equation}
Since each entry of $\mathbf{B}$ is $\pm 1$, we can prove that $\|\mathbf{B}\mathbf{w}\|_2\leq \sqrt{\Gamma}\|\mathbf{w}\|_1$ for any $(2^n-D)$-dimensional vector $\mathbf{w}$, and therefore
\begin{equation}
    \|\mathbf{y}-\mathbf{A}\boldsymbol{\coef}_\beta^{(k)}\|_2\leq \sqrt{\Gamma}(\eta+\|\mathbf{\bar{\boldsymbol{\coef}}_\beta^{(k)}}\|_1).
\end{equation}
This motivates us to solve a modified $\ell^1$-minimization problem
\begin{equation}
    \underset{\mathbf{z}\in\mathbb{C}^{D}}{\mathrm{minimize}}\|\mathbf{z}\|_1\quad \text{subject to }\|\mathbf{y}-\mathbf{A}\mathbf{z}\|_2\leq \sqrt{\Gamma}\widetilde{\eta},
\end{equation}
where $\widetilde{\eta}$ is any number such that $\widetilde{\eta}\geq  \eta+\|\mathbf{\bar{\boldsymbol{\coef}}_\beta^{(k)}}\|_1$. Solving this $\ell^1$-minimization problem yields a solution, which we denote by $\hat{\boldsymbol{\coef}}_\beta^{(k)}$. By Theorem~\ref{thm:compressed_sensing_main_thm} we have
\begin{equation}
    \|\hat{\boldsymbol{\coef}}_\beta^{(k)}-\hat{\boldsymbol{\coef}}_\beta\|_p \leq \frac{C_1}{M^{1-1/p}}\sigma_M(\boldsymbol{\coef}_\beta^{(k)})_1 + C_2 M^{1/p-1/2}\widetilde{\eta}.
\end{equation}
We let $\epsilon^{(k)}$ be a known upper bound of $\|\bar{\boldsymbol{\coef}}_\beta^{(k)}\|_1$, and considering all $3^k\log(M/\delta)$ randomly sampled indices $\beta$, we can therefore produce estimates $\hat{\coef}_P$ for all $k$-body coefficients satisfying
\begin{equation}
    \left(\sum_{P\in\mathbb{P}_n^{(k)}\setminus\{I\}} |\hat{\coef}_P-\coef_P|^p \right)^{1/p}\leq \epsilon + \frac{C_1 3^k\sigma_M(\boldsymbol{\coef})_1\log(M/\delta)}{M^{1-1/p}} + \frac{C_2 3^k\epsilon^{(k)}}{M^{1/2-1/p}},
\end{equation}
for universal constants $C_1,C_2>0$, with the same total evolution time as given in Theorem~\ref{thm:ham_learning_upper_bound}. Note that this error bound include both types of modeling errors.

\subsection{Computational complexity}
\label{sec:Computational_complexity}

In this section we will show that the classical post-processing takes time that is polynomial in $n$ and $\log(1/\epsilon_1)$ where $n$ is the number of qubits and $\epsilon_1$ is the allowed $\ell^1$-error on the coefficients. 
The classical post-processing comprises solving $\mathrm{poly}(n,\log(1/\epsilon_1))$ $\ell^1$-minimization problems as described in \eqref{eq:l1_minimization}, which we restate here: 
\begin{equation*}
    \operatornamewithlimits{minimize}_{\mathbf{z} \in \mathbb{C}^D}\|\mathbf{z}\|_1 \quad \text { subject to }\|\mathbf{A z}-\mathbf{y}\|_2 \leq \eta \sqrt{\Gamma}.
\end{equation*}
Here $\mathbf{A}$ is a $\Gamma\times D$ submatrix of the Hadamard matrix. An important feature that we will use later is that all entries of $\mathbf{A}$ are at most $1$ in absolute value (in fact they are all $\pm 1$).
Moreover, the admissible set of this optimization problem is non-empty, since the exact coefficient vector $\mathbf{x}$ that we want to recover satisfies $\|\mathbf{A}\mathbf{x}-\mathbf{y}\|_2\leq \sqrt{\Gamma}\eta$, and $\mathbf{x}$ is $M$-sparse whose entries are bounded by $1$ in absolute value.
We only need to show that approximately solving each of these $\ell^1$-minimization problems takes time $\mathrm{poly}(n,\log(1/\epsilon_1))$.

We will first demonstrate how this problem can be reformulated as a semidefinite programming (SDP) problem to build intuition. 
Because certain SDPs are pathological (for example, the optimal solution is exponentially large), this does not immediately imply that the solution is efficient.  
We therefore provide a self-contained analysis of the computational complexity by examining its dual~\eqref{eq:dual_problem}, and we rigorously prove that the classical computation requires weak polynomial time.

\subsubsection{Reformulation to an SDP}
We first make the substitution 
\[
    \mathbf{z} = \mathbf{z}^+ - \mathbf{z}^-,\quad \mathbf{v} = \mathbf{A}(\mathbf{z}^+ - \mathbf{z}^-) - \mathbf{y},
\]
where $\mathbf{z}^{\pm} = (z_1^{\pm},z_2^{\pm},\cdots,z_D^{\pm})$, and they are related to the $\mathbf{z}$ in \eqref{eq:l1_minimization} via $\mathbf{z} = \mathbf{z}^+ - \mathbf{z}^-$ and $z^+_i = \max\{z_i,0\}$, $z^-_i = -\min\{0,z_i\}$. 
Then \eqref{eq:l1_minimization} can be written as:
\begin{equation}
\begin{split}
    &\text{minimize } \mathbf{z}^+ + \mathbf{z}^-\\
    &\text{subject to: } \|\mathbf{v}\|_2 \leq \eta \sqrt{\Gamma}, \mathbf{z}^+ \geq 0, \quad \mathbf{z}^- \geq 0.
\end{split}
\end{equation}
Introduce the block matrix
    \begin{equation}
    \mathbf{W} = 
    \begin{pmatrix}
        \eta^2\Gamma & \mathbf{v}^* \\
        \mathbf{v} & 1 
    \end{pmatrix},
\end{equation}
where the first entry is a scalar, and the lower right block is an identity matrix. 
Then $\mathbf{W} \geq 0$ is equivalent to $\|\mathbf{v}\|_2 \leq \eta \sqrt{\Gamma}$ by a standard inequality on such block matrices.

The optimization becomes
\begin{equation}
\begin{split}
\text{minimize } &\mathbf{z}_+ + \mathbf{z}_-\\
\text{subject to: } &\mathbf{z}_+ \geq 0, \quad \mathbf{z}_- \geq 0,\\
&\mathbf{v} = \mathbf{A}(\mathbf{z}_+ - \mathbf{z}_-) - \mathbf{y},\\
&\mathbf{W} \geq 0, \quad \mathbf{W} = 
\begin{pmatrix}
\eta^2\Gamma & \mathbf{v}^* \\
\mathbf{v} & 1
\end{pmatrix},
\end{split}
\end{equation}
where the objective is linear and all the constraints are positive semi-definite constraints, meaning that this is an SDP.

In practice, semidefinite programs can be solved efficiently. However, to rigorously prove that \eqref{eq:l1_minimization} is solvable in polynomial time, a more careful analysis of the optimization error and the properties of the constraints, as in the following section, is required within the context of the Hamiltonian learning scenario we are considering.

\subsubsection{Self-contained analysis}
\label{sec:Self-contained analysis}

For the $\ell^1$-minimization problem in \eqref{eq:l1_minimization}, we consider a related problem in the following form
\begin{equation}
\label{eq:dual_problem}
    \text{minimize} \ \|\mathbf{Az}-\mathbf{y}\|_2\quad \text {subject to }  \|\mathbf{z}\|_1\leq \xi.
\end{equation}
This is equivalent to a convex quadratic program, 
\begin{equation}
\label{eq:convex_quadratic_program}
    \begin{aligned}
    &\text{minimize} \ \frac{1}{2}\|\mathbf{A}(\mathbf{z}^+-\mathbf{z}^-)-\mathbf{y}\|^2_2,\\
    &\text{subject to} \ \sum_i (z_i^+ + z_i^-)\leq\xi,\ \mathbf{z}^+\geq \mathbf{0},\ \mathbf{z}^-\geq \mathbf{0},
\end{aligned}
\end{equation}
where $\mathbf{z}^{\pm} = (z_1^{\pm},z_2^{\pm},\cdots,z_D^{\pm})$, and they are related to the $\mathbf{z}$ in \eqref{eq:dual_problem} via $\mathbf{z} = \mathbf{z}^+ - \mathbf{z}^-$ as before.
This convex quadratic program can be solved in time $\mathrm{poly}(\Gamma,D,\log(1/\epsilon'))$ so that the $\|\mathbf{Az}-\mathbf{y}\|_2$ obtained is $\epsilon'$ away from the optimal value \cite{Ye1989extension}.

We will first study the properties of the optimization problem \eqref{eq:dual_problem} and discuss how it helps us solve the original $\ell^1$ minimization problem \eqref{eq:l1_minimization}.
We denote by $\mathbf{z}(\xi)$ the solution to \eqref{eq:dual_problem} (which may not be unique, but any one of the solutions can be used) and by $f(\xi)$ the corresponding value of the objective function, i.e.
\[
    f(\xi) = \|\mathbf{Az}(\xi)-\mathbf{y}\|_2.
\]

We note that when $\xi$ is large enough, the constraint in \eqref{eq:dual_problem} will have no effect, and the optimization problem is solved by any $\mathbf{z}$ that minimizes $\|\mathbf{Az}-\mathbf{y}\|_2$, which is equivalent to $\mathbf{A}^\dag\mathbf{A}\mathbf{z}=\mathbf{A}^\dag \mathbf{y}$. The smallest value for $\xi$ for which this happens is
\begin{equation}
\label{eq:f_xi_plateau}
    \xi_0 = \min\{\|\mathbf{z}\|_1:\mathbf{A}^\dag\mathbf{A}\mathbf{z}=\mathbf{A}^\dag \mathbf{y}\}.
\end{equation}
\begin{lem}
\label{lem:monotonicity_of_f}
    Let $\mathbf{A}^+$ be the Moore–Penrose pseudoinverse of $\mathbf{A}$.
    \begin{itemize}
        \item[(i)] If $\xi\geq \xi_0$, then $f(\xi) = \sqrt{\mathbf{y}^\dag\mathbf{y}-\mathbf{y}^\dag \mathbf{A}\mathbf{A}^+\mathbf{y}}$.
        \item[(ii)] If $0\leq \xi\leq \xi_0$, then $\|\mathbf{z}(\xi)\|_1=\xi$ and $f(\xi)$ is a strictly decreasing function of $\xi$. 
        \item[(iii)] $f$ is Lipschitz continuous with Lipschitz constant $\sqrt{\Gamma}$. More precisely, for any $0\leq \xi_1\leq \xi_2$, we have
        \[
        f(\xi_2)\leq f(\xi_1)\leq f(\xi_2) + \sqrt{\Gamma}(\xi_2-\xi_1).
        \]
    \end{itemize}
\end{lem}

\begin{proof}
    For (i), we note that the constraint in \eqref{eq:dual_problem} no longer has any effect and we are solving an unconstrained quadratic optimization problem. 
    Solving this unconstrained quadratic optimization problem readily yields the expression for $f(\xi)$, which does not depend on $\xi$.

    For (ii), if $\|\mathbf{z}(\xi)\|_1<\xi$, then there exists an open neighborhood $U$ of $\mathbf{z}(\xi)$ such that $U\subset \{\mathbf{z}:\|\mathbf{z}\|_1\leq \xi\}$. 
    Because $\mathbf{z}(\xi)$ minimizes $\|A\mathbf{z}-\mathbf{y}\|^2$ within the open neighborhood $U$, the gradient of this objective function must be $\mathbf{0}$ at $\mathbf{z}(\xi)$. 
    This implies $\mathbf{A}^\dag\mathbf{A}\mathbf{z}(\xi)=\mathbf{A}^\dag \mathbf{y}$. 
    By \eqref{eq:f_xi_plateau}, we have $\|\mathbf{z}(\xi)\|_1\geq \xi_0$. Therefore $\xi_0\leq \|\mathbf{z}(\xi)\|_1 < \xi$, which contradicts $\xi\leq \xi_0$. 
    This proves $\|\mathbf{z}(\xi)\|_1=\xi$.

    Next we will prove that $f(\xi)$ is strictly decreasing for $0\leq \xi\leq \xi_0$. 
    If for $\xi_1,\xi_2\in [0,\xi_0]$ we have $\xi_1<\xi_2$ and $f(\xi_1)\leq f(\xi_2)$, then $\|\mathbf{A}\mathbf{z}(\xi_1)-\mathbf{y}\|_2\leq \|\mathbf{A}\mathbf{z}(\xi_2)-\mathbf{y}\|_2$, which implies that $\mathbf{z}(\xi_1)$ is a solution to the optimization problem \eqref{eq:dual_problem} with $\xi=\xi_2$. 
    Note that $\mathbf{z}(\xi_1)$ is in the interior of the admissible set of the this optimization problem, and therefore the gradient argument in the preceding paragraph applies, giving us $\mathbf{A}^\dag\mathbf{A}\mathbf{z}(\xi_1)=\mathbf{A}^\dag \mathbf{y}$, which results in the same contradiction. 
    This proves that $f(\xi_1)>f(\xi_2)$.

    For (iii), we only need to prove that $f(\xi_1)\leq f(\xi_2) + \sqrt{\Gamma}(\xi_2-\xi_1)$ for $\xi_1\leq \xi_2\leq \xi_0$. 
    We define $\mathbf{z}' = (\xi_1/\xi_2)\mathbf{z}(\xi_2)$. 
    Then by (ii) we have $\|\mathbf{z}'-\mathbf{z}_2\|_1 = (1-\xi_1/\xi_2)\|\mathbf{z}(\xi_2)\|_1 = \xi_2-\xi_1$. 
    Also, $\|\mathbf{z}'\|=\xi_1$ and $\mathbf{z}'$ is therefore in the admissible set of \eqref{eq:dual_problem} with $\xi=\xi_1$. Consequently
    \begin{equation}
    \label{eq:perturb_z_dual_problem}
        f(\xi_1)\leq \|\mathbf{A}\mathbf{z}'-\mathbf{y}\|_2\leq \|\mathbf{A}\mathbf{z}(\xi_2)-\mathbf{y}\|_2 + \|\mathbf{A}\mathbf{z}(\xi_2)-\mathbf{A}\mathbf{z}'\|_2 = f(\xi_2) +  \|\mathbf{A}(\mathbf{z}(\xi_2)-\mathbf{z}')\|_2.
    \end{equation}
    Note that $\|\mathbf{A}(\mathbf{z}(\xi_2)-\mathbf{z}')\|_2\leq \|\mathbf{A}\|_{1\to 2}\|\mathbf{z}(\xi_2)-\mathbf{z}'\|_1\leq \|\mathbf{A}\|_{1\to 2}(\xi_2-\xi_1)$, where
    \[
    \|\mathbf{A}\|_{1\to 2} := \max_{\|\mathbf{z}\|_1\leq 1} \|\mathbf{A}\mathbf{z}\|_2.
    \]
    Because all entries of $\mathbf{A}$ are $\pm 1$, we therefore have $\|\mathbf{A}\|_{1\to 2}\leq \sqrt{\Gamma}$. 
    Hence $\|\mathbf{A}(\mathbf{z}(\xi_2)-\mathbf{z}')\|_2\leq \sqrt{\Gamma}(\xi_2-\xi_1)$. 
    This combined with \eqref{eq:perturb_z_dual_problem} then results in the desired inequality.
\end{proof}

We then show that the above properties of $f(\xi)$ allow us to solve \eqref{eq:l1_minimization} by finding the right $\xi$. 
We consider the following cases for \eqref{eq:l1_minimization}: 
In the first case, $\sqrt{\Gamma}\eta < \sqrt{\mathbf{y}^\dag\mathbf{y}-\mathbf{y}^\dag \mathbf{A}\mathbf{A}^+\mathbf{y}} = f(\xi_0)$. In this scenario, the admissible set is empty, but we do not need to worry about this scenario because in our setup the exact coefficient vector is always in the admissible set. 
In the second case, $\sqrt{\Gamma}\eta  = f(\xi_0)$. 
The inequality constraint is then equivalent to $\mathbf{A}^\dag\mathbf{A}\mathbf{z} = \mathbf{A}^\dag \mathbf{y}$, an equality constraint. 
Therefore \eqref{eq:l1_minimization} can be reduced to a linear program and solved in weak polynomial time. 
Therefore, the only case we need to examine closely is the third case, where $\sqrt{\Gamma}\eta  > f(\xi_0)$.

\begin{lem}
    \label{lem:existence_of_good_xi}
    We assume $\sqrt{\Gamma}\eta  > \sqrt{\mathbf{y}^\dag\mathbf{y}-\mathbf{y}^\dag \mathbf{A}\mathbf{A}^+\mathbf{y}}$.
    If there exists $\mathbf{x}$ such that $\|\mathbf{A}\mathbf{x}-\mathbf{y}\|_2\leq \sqrt{\Gamma}\eta$, and $\|\mathbf{y}\|_2>\sqrt{\Gamma}\eta$, then there exists a unique $\xi^*\in [0,\min\{\|\mathbf{x}\|_1,\xi_0\}]$ such that $f(\xi^*)=\sqrt{\Gamma}\eta$, and $z(\xi^*)$ is a solution of \eqref{eq:l1_minimization}. 
\end{lem}

\begin{proof}
    Because $f(\|\mathbf{x}\|_1)\leq \|\mathbf{Ax}-\mathbf{y}\|_2\leq \sqrt{\Gamma}\eta$, and $f(0)=\|\mathbf{y}\|_2$, we have 
    \[
    \max\{f(\|\mathbf{x}\|_1),f(\xi_0)\}\leq \sqrt{\Gamma}\eta <f(0),
    \] 
    where we have used Lemma~\ref{lem:monotonicity_of_f} (i) which states that $f(\xi_0)=\sqrt{\mathbf{y}^\dag\mathbf{y}-\mathbf{y}^\dag \mathbf{A}\mathbf{A}^+\mathbf{y}}$.
    Because $f$ is continuous and strictly decreasing in $[0,\xi_0]$, we therefore have a unique $\xi^*$ such that $f(\xi^*)=\sqrt{\Gamma}\eta$, and
    \[
    0<\xi^*\leq \min\{\|\mathbf{x}\|_1,\xi_0\}.
    \]
\end{proof}

\begin{algorithm}
\label{alg:l1_minimization}
    \begin{algorithmic}
        \REQUIRE $\mathbf{A}\in\mathbb{C}^{\Gamma\times D}, \mathbf{y}\in\mathbb{C}^\Gamma, M\in\mathbb{N}_+,\eta>0$, error tolerance $\nu>0$.
        \STATE $\xi_L\leftarrow 0$, $\xi_R\leftarrow M$;
        \WHILE{$\xi_R-\xi_L>2\nu$}
            \STATE Solve the convex quadratic program \eqref{eq:convex_quadratic_program} with $\xi=(\xi_L+\xi_R)/2$ to obtain the minimum $f(\xi)^2/2$;
            \IF{$f(\xi)>\sqrt{\Gamma} \eta$}
            \STATE $\xi_L\leftarrow(\xi_L+\xi_R)/2$;
            \ELSE
            \STATE $\xi_R\leftarrow(\xi_L+\xi_R)/2$;
            \ENDIF
        \ENDWHILE
        \STATE $\tilde{\xi}=(\xi_L+\xi_R)/2$;
        \STATE Solve the convex quadratic program \eqref{eq:convex_quadratic_program} with $\xi=\tilde{\xi}$ to obtain the solution $\tilde{\mathbf{x}}^{\sharp}$;
        \ENSURE $\tilde{\mathbf{x}}^{\sharp}$.
    \end{algorithmic}
    \caption{Solving the $\ell^1$-minimization problem \eqref{eq:l1_minimization}.}
\end{algorithm}

We then show that solving the optimization problem \eqref{eq:dual_problem} with the $\xi^*$ in Lemma~\ref{lem:existence_of_good_xi} yields the solution to the $\ell^1$-minimization problem \eqref{eq:l1_minimization}. 

\begin{lem}
    \label{lem:good_xi_helps_solve_l1_minimization}
    With the same assumptions as Lemma~\ref{lem:existence_of_good_xi}, $\mathbf{z}(\xi^*)$ solves the $\ell^1$-minimization problem~\eqref{eq:l1_minimization}.
\end{lem}

\begin{proof}
    First note that because $f(\xi^*)=\sqrt{\Gamma}\eta$, $\mathbf{z}(\xi^*)$ is in the admissible set of \eqref{eq:l1_minimization}. 
    If there exists $\mathbf{z}'$ in the admissible set of \eqref{eq:l1_minimization} such that 
    \begin{equation}
    \label{eq:another_z_admissible_smaller_l1_norm}
        \|\mathbf{z}'\|_1 < \|\mathbf{z}(\xi^*)\|=\xi^*,
    \end{equation}
    where we have used $0<\xi^*\leq \xi_0$ and Lemma~\ref{lem:monotonicity_of_f} (ii) to show $\|\mathbf{z}(\xi^*)\|=\xi^*$,
    then 
    \[
    f(\|\mathbf{z}'\|_1)\leq \|\mathbf{A}\mathbf{z}'-\mathbf{y}\|\leq \sqrt{\Gamma}\eta = f(\xi^*),
    \]
    implying $\|\mathbf{z}'\|_1\geq \xi^*$ by Lemma~\ref{lem:monotonicity_of_f} (ii), which contradicts \eqref{eq:another_z_admissible_smaller_l1_norm}.
\end{proof}

Therefore solving the $\ell^1$-minimization problem \eqref{eq:l1_minimization} reduces to finding $\xi^*$, which can be done through a bisection procedure as described in Algorithm~\ref{alg:l1_minimization}. 
Each bisection step involves solving \eqref{eq:dual_problem} for a certain $\xi$, which is equivalent to the convex quadratic program \eqref{eq:convex_quadratic_program}. 
In the bisection process, we start from an interval $[0,M]$, in each iteration compute $f(\xi)$ where $\xi$ is the mid-point of the interval, thereby deciding whether $\xi^*$ is in the right half of the interval or the left half. 
This allows us to get a new interval with half the size but still contains $\xi^*$, thereby proceeding to the next iteration. 
When the interval size is $2\nu$ we can take the mid-point $\tilde{\xi}^*$ as the estimate for $\xi^*$, which ensures
\[
|\tilde{\xi}^*-\xi^*|\leq \nu.
\]
The total number of iterations needed is at most $\lceil \log_2(M/\nu)\rceil$.

However, our ultimate goal is to recover the original coefficient vector $\mathbf{x}$, and we need to know how small $\nu$ needs to be in order for the approximate solution $\tilde{\mathbf{x}}^\sharp$ (obtained by solving \eqref{eq:convex_quadratic_program} with $\xi=\tilde{\xi}^*$) to be close to $\mathbf{x}$. 
Next we will provide an answer to this question.

Note that solving \eqref{eq:convex_quadratic_program} with $\xi=\tilde{\xi}^*$ is equivalent to solving the $\ell^1$-minimization problem \eqref{eq:l1_minimization} but with $\eta$ replaced with some $\widetilde{\eta}$. 
This $\widetilde{\eta}$ can be identified by $f(\tilde{\xi}^*)=\sqrt{\Gamma}\widetilde{\eta}$. 
Therefore by the Lipschitz continuity of $f$ (Lemma~\ref{lem:monotonicity_of_f} (iii)) we have
\[
|\sqrt{\Gamma}\widetilde{\eta}-\sqrt{\Gamma}\eta| =|f(\tilde{\xi}^*)-f(\xi^*)|\leq \sqrt{\Gamma}|\tilde{\xi}^*-\xi^*|,
\]
which leads to
\[
|\widetilde{\eta}-\eta|\leq |\tilde{\xi}^*-\xi^*|\leq \nu.
\]
From the above we can see that in the end we obtain an approximate solution $\tilde{\mathbf{x}}^\sharp$ that corresponds to the exact solution of \eqref{eq:l1_minimization} but with $\eta$ replaced with some $\widetilde{\eta}$ satisfying $\widetilde{\eta}\leq \eta + \nu$. 
We can therefore use Theorem~\ref{thm:compressed_sensing_main_thm} to obtain the following guarantee about the closeness to $\mathbf{x}$:

\begin{thm}[Compressed sensing with approximate solution]
\label{thm:compressed_sensing_approx_solution}
    Let $D=\sum_{l=0}^k\binom{n}{l}$.
    Let $\mathbf{A} \in \mathbb{C}^{\Gamma \times D}$ be a matrix whose rows are independently randomly sampled from the rows of $\mathbf{H}^{(k)}$ with replacement.
    Let $\delta_{\mathrm{CS}}\in(0,1)$ and let $M>0$ be an integer. If
    \begin{equation}
        \Gamma \geq CM\max\{\ln^2(M)\ln(M\ln(D))\ln(D),\ln(1/\delta_{\mathrm{CS}})\},
    \end{equation}
    then, with probability at least $1-\delta_{\mathrm{CS}}$, the following statement holds for every $M$-spase $\mathbf{x} \in \mathbb{C}^D$ whose entries are bounded by $1$ in absolute value. Let noisy samples $\mathbf{y}=\mathbf{A x}+\mathbf{e}$ be given with $\|\mathbf{e}\|_2\leq \eta \sqrt{\Gamma}$,
        and let $\tilde{\mathbf{x}}^{\sharp}$ be the approximate solution of the $\ell^1$-minimization problem \eqref{eq:l1_minimization} produced by Algorithm~\ref{alg:l1_minimization} with the error tolerance on $\xi^*$ chosen to be $\nu>0$.
        Then
        \begin{equation}
             \left\|\mathbf{x}-\tilde{\mathbf{x}}^{\sharp}\right\|_p \leq C_1 M^{1/p-1/2} (\eta+\nu), \quad 1\leq p \leq 2.
        \end{equation}
        Both constants $C, C_1$ are universal. The approximate solution $\tilde{\mathbf{x}}^{\sharp}$ is obtained through Algorithm~\ref{alg:l1_minimization} which runs in time $\mathrm{poly}(n,\Gamma,\log(1/\nu))$.
\end{thm}

\section{Lower bounds}
\label{sec:lower_bounds}
In this section we provide a lower bound for the Hamiltonian learning task we consider. 
Specifically we consider the dependence on the accuracy $\epsilon$ and the number of Pauli terms $M$. 
The adaptive experiments are modeled as in \cite{huang2022foundations}, where all experiments form a tree such that the outcome of the experiment at a vertex determines which child leaf to move to, thus determining the next experiment to perform.
Our main tool for the total evolution time lower bound is Assouad's lemma \cite{yu1997assouad}, which provides a lower bound on the achievable $\ell^1$-error given two prerequisites: 
(1) an estimate of how hard it is to distinguish two output probability distributions if they come from two Hamiltonians that differ slightly, and 
(2) a lower bound on the penalty in $\ell^1$-error if such a pair of Hamiltonians are not correctly distinguished. 
The second prerequisite is easy to fulfill, as discussed in the proof of Theorem~\ref{thm:lower_bound}. 
For the first prerequisite, we will follow \cite{HuangTongFangSu2023learning} to capture the difficulty of correctly distinguishing the output probability distributions by induction on the tree of adaptive experiments.
In the following, we first characterize the model of quantum learning experiments following \cite{huang2022foundations}. 
Details of the proofs are in Appendix~\ref{sec:Technical lemmas for the lower bound}. 

\subsection{Model of quantum experiments}
We consider the scenario of learning the unknown Hamiltonian $H$ from dynamics. Consider the time evolution operator of $H$ parameterized by time $t$:
\begin{equation}
    U(t) = e^{-iHt},
\end{equation}
A learning agent can have access to $U(t)$ by quantum experiments without knowing $H$. 
The learning algorithm includes the scenario where each quantum experiment is chosen adaptively based on past measurement outcomes.
Following \cite{huang2022foundations}, we first define a single experiment:

\begin{defn}[An ideal experiment]
\label{defn:ideal_experiment}
    Given a $n$-qubit unitary $U(t)=e^{-i H t}$ parameterized by time $t$, with $H$ being the unknown Hamiltonian. A single ideal experiment $E^{(0)}$ is specified by:
    \begin{enumerate}
        \item[1.] an arbitrary $n^{\prime}$-qubit initial state $\ket{\psi_0} \in \mathbb{C}^{2^{n^{\prime}}}$ with an integer $n^{\prime} \geq n$,
        \item[2.] an arbitrary POVM $\mathcal{F}=\left\{M_i\right\}_i$ on an $n^{\prime}$-qubit system,
        \item[3.] an $n^{\prime}$-qubit unitary of the following form,
        \begin{equation}
        \label{eq:unitary_operator_for_time_evolution}
            U_{\mathrm{exp}} = U_{K+1}\left(U\left(t_K\right) \otimes I\right) U_K \ldots U_3\left(U\left(t_2\right) \otimes I\right) U_2\left(U\left(t_1\right) \otimes I\right) U_1,
        \end{equation}
        for some arbitrary integer $K$, arbitrary evolution times $t_1, \ldots, t_K \in \mathbb{R}$, and arbitrary $n^{\prime}$-qubit unitaries $U_1, \ldots, U_K, U_{K+1}$. Here $I$ is the identity unitary on $n^{\prime}-n$ qubits subsystem.
    \end{enumerate} 
    A single run of $E^{(0)}$ returns an outcome from performing the POVM $\mathcal{F}=\{M_i\}_i$ on the state $\mathbf{U}\ket{\psi_0}$.
    The evolution time of the experiment is defined as
    \begin{equation}
        t\left(E^{(0)}\right) \coloneqq \sum_{k=1}^Kt_k.
    \end{equation}
\end{defn}

Next, we will take noise into account. 
A small SPAM error is unavoidable in any quantum experiment, and being robust against it is an important consideration in useful algorithms for characterizing and benchmarking quantum systems \cite{KimmelLowYoder2015robust, huang2022foundations, knill2008randomized, magesan2011scalable, erhard2019characterizing, harper2020efficient, Elben_2022}.
We model the SPAM error by a global depolarizing channel 
\begin{equation}
    \label{eq:global_depolarizing}
    \mathcal{D}: \rho\mapsto \mathcal{D}(\rho) = (1-\gamma)\rho + \frac{\gamma\tr(\rho)}{2^{n'}}I,
\end{equation}
which is applied to the system right before the POVM $\mathcal{F}$. This is equivalent to changing the POVM to 
\begin{equation}
    \label{eq:modified_POVM_by_SPAM}
    \mathcal{F}^{(\gamma)}=\left\{(1-\gamma) M_i+\gamma \operatorname{tr}\left(M_i\right)\left(I / 2^{n^{\prime}}\right)\right\}_i.
\end{equation}
The noise strength is measured by the constant $\gamma$, which can take any value in $(0,1)$.

In proving the lower bound, we are allowed to only consider SPAM error in the form of the global depolarizing noise applied in the way described above, because any algorithm that is robust to SPAM error in general must be robust to this type of SPAM error in particular. 
We will then define a single experiment with measurement noise as follows:

\begin{defn}[An experiment with measurement noise]
\label{defn:single_experiment_with_measurement_noise}
A single experiment $E^{(\gamma)}$ with measurement noise $\gamma\in(0,1)$ is identical to a single ideal experiment $E^0$ except for the POVM, which is $\mathcal{F}^{(\gamma)}$ given in \eqref{eq:modified_POVM_by_SPAM} instead of $\mathcal{F}=\{M_i\}_i$.
\end{defn}

Note that a lower bound for learning with measurement noise naturally also holds for learning with both state-preparation and measurement errors, since the latter is a more difficult task. 
Therefore for simplicity we only consider measurement noise in this section.

With the definition of a single experiment with measurement noise, we can now formally define a learning algorithm with total evolution time $T$ as follows, where we again adopt the definition in \cite{HuangTongFangSu2023learning}.

\begin{defn}[Learning algorithm with bounded total evolution time]
\label{defn:learning_alg_with_bdd_total_evolution_time}
Consider $T>0,0.5>\gamma>0$. 
A learning algorithm with total evolution time $T$ and measurement noise $\gamma$ can obtain measurement outcomes from an arbitrary number of experiments $E_1^{(\gamma)}, E_2^{(\gamma)}, \ldots$ as long as
\begin{equation}
    \sum_i t\left(E_i^{(\gamma)}\right) \leq T
\end{equation}
The parameters specifying each experiment $E_i^{(\gamma)}$ can depend on the measurement outcomes from previous experiments $E_1^{(\gamma)}, \ldots, E_{i-1}^{(\gamma)}$.
\end{defn}

\subsection{Lower bound of the total evolution time in Hamiltonian learning}
With the definition of the learning algorithm and the possible sets of experiments, we now present a fundamental lower bound on the total evolution time for any learning algorithm that tries to learn an unknown $n$-qubit Hamiltonian consisting of $M$ Pauli terms from real-time evolution. The theorem is stated as follows:

\begin{thm}
\label{thm:lower_bound}
    Given integers $n$ and $M$, real numbers $\epsilon, \delta \in(0,1)$, and a set $\left\{P_1, P_2, \ldots, P_M\right\} \in \mathbb{P}_n\setminus\{I\}$ representing the $M$ $n$-qubit Pauli terms in the Hamiltonian, we consider any learning algorithm with a total evolution time $T$ and a constant measurement noise $\gamma\in(0,0.5)$ as defined in Definition~\ref{defn:learning_alg_with_bdd_total_evolution_time}.
    Suppose that such a learning algorithm satisfies that for an $n$-qubit Hamiltonian $H=\sum_{a=1}^M \coef_aP_a$ with any unknown parameters $\abs{\coef_a}\leq 1$, after multiple rounds of noisy experiments, the algorithm can estimate any $\boldsymbol{\coef}= \left(\coef_1,\ldots,\coef_M\right)$ to $\epsilon_1$-error in the $\ell^1$-norm in expectation value averaged over experimental outcomes. Then, 
    \begin{equation}
    \label{eq:total_evolution_time_lower_bound}
        T\geq \frac{M}{\epsilon_1 e \log(1/\gamma)}.
    \end{equation}
\end{thm}
Notice that the lower bound in Theorem~\ref{thm:lower_bound} is for learning $n$-qubit Hamiltonians with fixed $M$ terms. It naturally implies a lower bound for learning $n$-qubit $M$-term Hamiltonians without prior knowledge of which $M$ terms are in the Hamiltonian. Moreover, by restricting the terms in the Hamiltonian to be $k$-body, this lower bound also holds for the $k$-body, $M$-term, $n$-qubit Hamiltonian learning problem we discussed in the previous sections.

We begin by considering how well a single noisy experiment can distinguish two Hamiltonians.
\begin{lem}[TV for one experiment, Lemma~30 \cite{HuangTongFangSu2023learning}]
\label{lem: TV for one experiment}
Consider a noisy experiment as defined in Definition~\ref{defn:single_experiment_with_measurement_noise}.
Let $p(i)$ and $p^{\prime}(i)$ be the probability of obtaining the measurement outcome $i$ by performing 
the POVM $\mathcal{F}^{(\gamma)}$ when the unknown Hamiltonian is $H$ and $H'$ respectively.
Then
\begin{equation}
    \operatorname{TV}\left(p, p^{\prime}\right) \leq(1-\gamma) \min \left(\|H-H'\| \cdot t\left(E^{(\gamma)}\right), 1\right),
\end{equation}
where $t\left(E^{(\gamma)}\right)=\sum_{k=1}^K\left|t_k\right|$ is the total evolution time in this single experiment $E^{(\gamma)}$. Here $\operatorname{TV}$ denotes the total variation distance between probability distributions. 
\end{lem}

\begin{proof}
    We define 
    \begin{equation} 
    \begin{split}
    \ket{\psi} &= U_{K+1}\left(U\left(t_K\right) \otimes I\right) U_K \ldots U_3\left(U\left(t_2\right) \otimes I\right) U_2\left(U\left(t_1\right) \otimes I\right) U_1\ket{\psi_0},\\
    \ket{\psi'} &= U_{K+1}\left(U'\left(t_K\right) \otimes I\right) U_K \ldots U_3\left(U'\left(t_2\right) \otimes I\right) U_2\left(U'\left(t_1\right) \otimes I\right) U_1\ket{\psi_0}
    \end{split}
    \end{equation}
    By the triangle inequality and telescoping sum, we have the following upper bound on the trace distance,
    \begin{equation}
    \label{eq:TV upper bound}
        \|\ketbra{\psi}{\psi}-\ketbra{\psi'}{\psi'}\|_1\leq\sum_{k=1}^K\|U(t_k)-U'(t_k)\|_{\diamond} \leq 2\|H-H'\|\cdot\sum_{k=1}^K\abs{t_k} = 2\|H-H'\|\cdot t\left(E^{(\gamma)}\right),
    \end{equation}
    where we have used Lemma~\ref{lem:diamond_distance_unitary_channels} in the second inequality. 
    Because the ideal POVM measurement $\mathcal{F}=\left\{M_i\right\}_i$ is a quantum channel, we have
    \begin{equation}
        \frac{1}{2}\sum_i\abs{\bra{\psi}M_i\ket{\psi}-\bra{\psi'}M_i\ket{\psi'}}\leq \frac{1}{2} \|\ketbra{\psi}{\psi}-\ketbra{\psi'}{\psi'}\|_1\leq \min\left(\|H-H'\|\cdot t\left(E^{(\gamma)}\right),1\right).
    \end{equation}
    Now we take measurement noise into account. 
    Since the noisy POVM $\mathcal{F}^{(\gamma)}=\left\{\Tilde{M}_i\right\}$ where
    \[
    \Tilde{M}_i= (1-\gamma) M_i+\gamma \operatorname{tr}\left(M_i\right)\left(I / 2^{N^{\prime}}\right)\,,
    \]
    the total variation distance between the measurement outcome distribution is
    \begin{equation}
        \begin{split}
        \operatorname{TV}\left(p, p^{\prime}\right) &= \frac{1}{2}\sum_i\abs{\bra{\psi}\Tilde{M}_i\ket{\psi}-\bra{\psi'}\Tilde{M}_i\ket{\psi'}}\\
            &= \frac{1}{2}(1-\gamma)\sum_i\abs{\bra{\psi}M_i\ket{\psi}-\bra{\psi'}M_i\ket{\psi}} \\
            &\leq (1-\gamma)\min\left(\|H-H'\|\cdot t\left(E^{(\gamma)}\right),1\right).
        \end{split}
    \end{equation}
\end{proof}

In order to include possible adaptivity in the choice of experiments, we consider the rooted tree representation $\mathcal{T}$ described in \cite{chen2022exponential,huang2022quantum}, which was also used in \cite{HuangTongFangSu2023learning}. 
Each node in the tree corresponds to the sequence of measurement outcomes the algorithm has seen so far. 
We can also think of the node as the memory state of the algorithm. 
At each node $v$, the algorithm runs a single experiment $E_v^{(\gamma)}$ as defined in Definition~\ref{defn:single_experiment_with_measurement_noise} with the measurement noise specified by $\gamma$ and evolution time specified by $T_v$.

For each experiment $E_v^{(\gamma)}$, it produces a measurement outcome $i\in{1,2,\ldots,L_v}$. 
The outcome will determine which child node to move the algorithm to from node $v$. 
When the algorithm arrives at a leaf $\ell$, the algorithm stops, with $\ell$ denoting the outcome of the experiments. 
Note that because of the tree structure, a leaf node $\ell$ uniquely determines the path leading from the root to it. 
By considering the rooted tree representation and allowing the algorithm to depend on each node in the tree adaptively, we cover all possible learning algorithms that can adaptively choose the experiment that it runs based on previous measurement outcomes.

For each node $v$ on tree $\mathcal{T}$, we denote $p^{(\mathcal{T})}(\ell)$ and $p'^{(\mathcal{T})}(\ell)$ as the probabilities of arriving at the leaf node $\ell$ in the experiments when the unknown Hamiltonians are $H$ and $H'$ respectively, and the algorithm begins from the root of $\mathcal{T}$. 
We denote the resulting distribution over leaf nodes by $p^{(\mathcal{T})}$ and $p'^{(\mathcal{T})}$ respectively. 
The total evolution time $T$ can be written as:
\begin{equation}
    T = \sum_{v\in \mathrm{p}(r,\ell)}T_v,
\end{equation}
where $\mathrm{p}(r,\ell)$ is the path from the root $r$ to a leaf $\ell$ of $\mathcal{T}$.
We will provide an upper bound of the TV distance between the probability distributions $p^{(\mathcal{T})}$ and $p'^{(\mathcal{T})}$ with a given $T$.

\begin{lem}[TV for multiple experiments]
\label{lem: TV for multiple experiments}
    Consider a rooted tree representation $\mathcal{T}$ for a learning algorithm with total evolution time $T$ and measurement noise $\gamma \in(0,0.5)$. 
    We have
    \begin{equation}
        \operatorname{TV}\left(p^{(\mathcal{T})}, p'^{(\mathcal{T})}\right) \leq 1-\gamma^{\|H-H'\|\cdot T},
    \end{equation}
    which is an upper bound for the total variation of the outcomes under multiple experiments.
\end{lem}

\begin{proof}
    The proof is almost identical to the one for \cite[Lemma 31]{HuangTongFangSu2023learning}, which is by induction on the tree $\mathcal{T}$. 
    The only change is to use Lemma~\ref{lem: TV for one experiment} to upper bound the TV distance that can be generated in a single noisy experiment. 
\end{proof}

The above theorem tells us how much the output probability distribution can be changed when the Hamiltonian is perturbed. 
It therefore quantifies how well we can distinguish a pair of Hamiltonians given limited total evolution time. 
With this tool we are ready to prove the lower bound result in Theorem~\ref{thm:lower_bound}.

The main tool we are going to use next is Assouad's lemma \cite{yu1997assouad}. 
We consider a set of probability distributions $P$ each of which correspond to a parameter $\theta(P)\in\mathfrak{D}$, where the space $\mathfrak{D}$ can be equipped with certain pseudo-distances.
We restate the lemma here:

\begin{lem}(Assouad's Lemma\cite[Lemma 2]{yu1997assouad})
\label{lem:assouads_lemma}
Let $M \geq 1$ be an integer and let $\mathcal{F}_M=$ $\left\{P_\tau: \tau \in\{-1,1\}^M\right\}$ contain $2^M$ probability measures. 
Write $\tau \sim \tau^{\prime}$ if $\tau$ and $\tau^{\prime}$ differ in only one coordinate, and write $\tau \sim_j \tau^{\prime}$ when that coordinate is the $j$th. 
Suppose that there are $M$ pseudo-distances on $\mathfrak{D}$ such that for any $x, y \in \mathfrak{D}$
\begin{equation}
d(x, y)=\sum_{j=1}^M d_j(x, y)
\end{equation}
and further that, if $\tau \sim_j \tau^{\prime}$,
\begin{equation}
d_j\left(\theta\left(P_\tau\right), \theta\left(P_{\tau^{\prime}}\right)\right) \geq \alpha_M.
\end{equation}
Then
\begin{equation}
\max _{P_\tau \in \mathcal{F}_M} E_\tau d\left(\hat{\theta}, \theta\left(P_\tau\right)\right) \geq M \cdot \frac{\alpha_M}{2} \min \left\{\left\|P_\tau \wedge P_{\tau^{\prime}}\right\|: \tau \sim \tau^{\prime}\right\},
\end{equation}
where $\|\cdot\wedge\cdot\| = 1-\mathrm{TV}(\cdot,\cdot)$ is the affinity of two probability distributions.
\end{lem}

We will the apply this lemma to prove Theorem~\ref{thm:lower_bound}.

\begin{proof}[Proof of Theorem~\ref{thm:lower_bound}]
    In the Hamiltonian learning scenario, each $\tau$ represents a coefficient vector $\boldsymbol{\coef}= \left(\coef_1,\ldots,\coef_M\right)$, where $\coef_i = h\tau_i\in\{-h,h\}$ and $h\in(0,1)$. $P_\tau$ is the probability distribution $p^{\mathcal{T}}$ over leaf nodes of $\mathcal{T}$ corresponding to adaptive noisy experiments with the Hamiltonian represented by coefficient vector $\tau$. 
    The parameter $\theta(P_\tau) = \boldsymbol{\coef}$. 
    The distance on $\mathfrak{D}$ is set to be the $\ell^1$-norm distance between the corresponding coefficient vectors. 
    For adaptive experiments with Hamiltonian coefficients $\boldsymbol{\coef}$ and $\boldsymbol{\coef}'$, with output distributions $P_\tau=p^{\mathcal{T}}$ and $P'_\tau=p'^{\mathcal{T}}$, we define
    \begin{equation}
        d(\theta(P_\tau), \theta(P'_\tau))=d(\boldsymbol{\coef},\boldsymbol{\coef}')=\sum_{j=1}^M \abs{\coef_j-\coef_j'}.
    \end{equation}
    Using the notation $\sim_j$ introduced in Lemma~\ref{lem:assouads_lemma}, if $\tau \sim_j \tau'$, we have $\alpha_M=d(\boldsymbol{\coef},\boldsymbol{\coef}')=d_j(\boldsymbol{\coef},\boldsymbol{\coef}')=2h$, where $\boldsymbol{\coef}=h\tau$ and $\boldsymbol{\coef}'=h\tau'$. 
    For the affinity $\|P_\tau\wedge P'_\tau\|$, by Lemma~\ref{lem: TV for multiple experiments} we have
    \[
        \|P_\tau\wedge P'_\tau\| = 1-\mathrm{TV}(P_\tau,P'_\tau)\geq \gamma^{\|H-H'\|T}=\gamma^{2hT}.
    \]

    With the above parameters, we can directly apply Assouad's Lemma in Lemma~\ref{lem:assouads_lemma} to obtain
    \begin{equation}
        \max _{P_\tau \in \mathcal{F}_M}\mathop{\mathbb{E}}_{\tau}\|\boldsymbol{\coef}-\hat{\boldsymbol{\coef}}\| \geq 2Mh\gamma^{2hT},
    \end{equation}
    where $\hat{\boldsymbol{\coef}}$ is the estimate produced by any sequence of adaptive noisy experiments with total evolution time $T$. 
    Note that this bound holds for any $h\in(0,1)$, and we can therefore choose $h=1/(2T\log(1/\gamma))$, and we will have
    \begin{equation}
        \max _{P_\tau \in \mathcal{F}_M}\mathop{\mathbb{E}}_{\tau}\|\boldsymbol{\coef}-\hat{\boldsymbol{\coef}}\| \geq\frac{M}{Te\log(1/\gamma)}.
    \end{equation}
    Therefore if we want to obtain coefficient estimates for Hamiltonians of the form $H=\sum_{a=1}^M\coef_a P_a$ that is accurate up to error $\epsilon_1$ in $\ell^1$-distance, we need $T$ as given in \eqref{eq:total_evolution_time_lower_bound}.
\end{proof}

\section{Operational interpretation of error metrics}
\label{sec:operational_interpretation}

In the previous sections we have discussed the method to learn the coefficients of a Hamiltonian $H$ and the corresponding total evolution time lower bound. 
The ultimate goal of Hamiltonian learning is to use the learned Hamiltonian to predict the properties of the quantum system. 
In this section we will consider this task, and thereby provide an operational interpretation of the error metrics.

We consider the task of predicting observation expectation values at a given time $t$.
Suppose that the exact Hamiltonian is $H=\sum_{P}\coef_P P$ and the learned Hamiltonian is $H'=\sum_P\coef_P' P$. 
Here, because we fix the global phase by enforcing $\coef_I=0$, these two Hamiltonians are both traceless. 
Starting from a known quantum state $\rho_0$, the exact state $\rho(t)$ and the state $\rho'(t)$ we predict will be
\begin{equation}
    \rho(t) = e^{-iHt}\rho_0 e^{iHt},\quad \rho'(t) = e^{-iH't}\rho_0 e^{iH't}.
\end{equation}
If we focus on a specific observation $O$, then the expectation values are $\Tr[O\rho(t))]$ and $\Tr[O\rho'(t)]$. 
We will then discuss what the $\ell^1$- and $\ell^2$-error bounds imply for the error $\abs{\Tr[O\rho'(t)]-\Tr[O\rho(t)]}$, which reflects how well we can make predictions based on the learned Hamiltonian. 
Such considerations of the observable, or alternatively the initial state, have featured prominently in works on quantum simulation and related error analysis \cite{chen2021concentration,ZhaoZhouEtAl2022hamiltonian,BornsWeil2022uniform,SuHuangCampbell2021nearly,AnFangLin2021time,FangVilanova2023observable,TongAlbertEtAl2022provably,CaiTongPreskill2024}.

Here we will first state our conclusions: the $\ell^1$-error bound provides a guarantee that we can accurately predict any observable $O$ whose spectral norm $\|O\|$ is not too large, for any initial state (see \eqref{eq:expectation_val_l1_bound}); the $\ell^2$-error bound guarantees that we can accurately predict any observable $O$ whose Frobenius norm $\|O\|_2$ is not too large (these are typically low-rank operators) for an initial state drawn from a Haar distribution or a 2-design (see \eqref{eq:expectation_val_l2_bound}).

\subsection{Operational interpretation of the $\ell^1$-norm}
We first consider the $\ell^1$-norm distance between the coefficient vectors
\begin{equation}
    \|\boldsymbol{\coef}-\boldsymbol{\coef}'\|_1 = \sum_P \abs{\coef_P - \coef_P'}.
\end{equation}
By Schrodinger's equation, we have the ordinary differential equation for $\rho'$:
\begin{equation}
    \frac{\dd}{\dd t}\rho'(t) = -i[H',\rho'(t)] = -i[H,\rho'(t)] -i[H'-H,\rho'(t)].
\end{equation}
By Duhamel's principle, we have
\begin{equation}
    \rho'(t) = \rho(t) - i \int_0^t\dd s e^{-iH(t-s)}[H'-H, \rho'(s)] e^{iH(t-s)}.
\end{equation}
Therefore, the trace distance between $\rho(t)$ and $\rho'(t)$ can be bounded as:
\begin{equation}
    \|\rho'(t)-\rho(t)\|_{1} \leq  2t\|H'-H \|_{\infty}.
\end{equation}
Here $\|\cdot\|_p$ denotes the Schatten $p$-norm. Specifically, $p=1$ corresponds to the trace norm, and $p=\infty$ corresponds to the spectral norm. The spectral norm of the difference between $H$ and $H'$ can be bounded by
\begin{equation}
    \|H-H'\|_\infty \leq \|\boldsymbol{\coef}-\boldsymbol{\coef}'\|_1.
\end{equation}
Thus,
\begin{equation}
\label{eq:bounding_trace_norm_with_l1_norm}
    \|\rho(t)-\rho'(t)\|_1\leq 2t\|\boldsymbol{\coef}-\boldsymbol{\coef}'\|_1.
\end{equation}

The trace norm distance between $\rho(t)$ and $\rho(t)'$ is upper bounded by the $\ell^1$-norm distance between the two coefficient vectors, scaling linearly with evolution time.

Consider the expectation value of measuring $\rho(t)$ and $\rho'(t)$ with observable $O$,
\begin{equation}
\label{eq:difference_between_expval}
    \abs{\expval{O}_{\rho(t)} - \expval{O}_{\rho'(t)}} = \abs{\Tr[O(\rho(t)-\rho'(t))]} \leq \|O\|_\infty\|\rho(t)-\rho'(t)\|_1.
\end{equation}
The second inequality comes from the H\"older's inequality for Schatten norms. Using the bound in \eqref{eq:bounding_trace_norm_with_l1_norm}, the difference between the expectation values can be bounded by
\begin{equation}
\label{eq:expectation_val_l1_bound}
    \abs{\expval{O}_{\rho(t)} - \expval{O}_{\rho'(t)}} \leq 2t\|O\|_\infty \|\boldsymbol{\coef}-\boldsymbol{\coef}'\|_1.
\end{equation}
Note that the above is true for all $\rho_0$ and $O$, and therefore the $\ell^1$-error bound provides us with a worst case guarantee on how well we can predict expectation values of observables.

\subsection{Operational interpretation of the $\ell^2$-norm}
We now consider the $\ell^2$-norm distance between the coefficient vectors
\begin{equation}
    \|\boldsymbol{\coef}-\boldsymbol{\coef}'\|_2 = \sqrt{\sum_a \left(\coef_a - \coef_a'\right)^2}.
\end{equation}
Parseval's identity provides that the $\ell^2$-norm of the distance between the coefficient vectors equals the Frobenious norm of $H-H'$:
\begin{equation}
    \frac{1}{\sqrt{2^n}}\|H-H'\|_2 = \sqrt{\sum_i (\coef_i-\coef_i')^2}.
\end{equation}
Following \cite{bluhm2024hamiltonian}, we consider the time evolution channel of $H$ and $H'$
\begin{equation}
    \mathcal{U}_t(\rho) = e^{-iHt}\rho e^{iHt},\quad \mathcal{U}'_t(\rho) = e^{-iH't}\rho e^{iH't}
\end{equation}
and the normalized Choi states of $\mathcal{U}_t$ and $\mathcal{U}'_t$ as
\begin{equation}
        \mathcal{C}(\mathcal{U}_t) = (\mathcal{U}_t\otimes\mathcal{I})(\Omega)\quad \mathcal{C}(\mathcal{U}'_t) = (\mathcal{U}'_t\otimes\mathcal{I})(\Omega),
\end{equation}
where $\Omega = \ketbra{\Omega}{\Omega}$ is the maximally entangled state, and $\mathcal{I}$ is the identity channel. 
We introduce the distance metric $\mathcal{D}(\mathcal{U}_t,\mathcal{U}_t') = \frac{1}{\sqrt{2}}\|\mathcal{C}(\mathcal{U}_t)-\mathcal{C}(\mathcal{U}_t')\|_2$.
\begin{lem}
    \label{lem:difference_unitary_channels_choi_states}
    The following inequality holds
    \[
    D\left(\mathcal{U}_t, \mathcal{U}'_t\right)\leq \frac{t}{\sqrt{2^n}}\|H-H'\|_2.
    \]
\end{lem}
\begin{proof}
    First note that
    \[
    (e^{-iHt}\otimes I)\ket{\Omega} = (I\otimes e^{-iH^\top t})\ket{\Omega}.
    \]
    Then we have
    \[
    \begin{aligned}
        [(H-H')\otimes I,(e^{-iHt}\otimes I)\Omega (e^{iHt}\otimes I)] &=[(H-H')\otimes I,(I\otimes e^{-iH^\top t})\Omega (I\otimes e^{iH^\top t})] \\
        &= (I\otimes e^{-iH^\top t})[(H-H')\otimes I,\Omega] (I\otimes e^{iH^\top t}),
    \end{aligned}
    \]
    where we have used the fact that $I\otimes e^{-iH^\top t}$ commute with $(H-H')\otimes I$.
    Therefore
    \[
    \|[(H-H')\otimes I,(e^{-iHt}\otimes I)\Omega (e^{iHt}\otimes I)]\|^2_2 = \|[(H-H')\otimes I,\Omega]\|_2^2 = \frac{2}{2^n}\|H-H'\|^2_2.
    \]
    Then because
    \[
    \begin{aligned}
        \mathcal{U}'_t(\Omega) -  \mathcal{U}_t(\Omega) 
        &= (e^{-iH't}\otimes I)\Omega (e^{iH't}\otimes I)- (e^{-iHt}\otimes I)\Omega (e^{iHt}\otimes I) \\
        &= -i\int_0^t e^{-iH's}[(H-H')\otimes I,(e^{-iHs}\otimes I)\Omega (e^{iHs}\otimes I)]e^{iH's} \dd s,
    \end{aligned}
    \]
    we have
    \[
    \begin{aligned}
        D\left(\mathcal{U}_t, \mathcal{U}'_t\right) &\leq \frac{1}{\sqrt{2}}\int_0^t \|[(H-H')\otimes I,(e^{-iHs}\otimes I)\Omega (e^{iHs}\otimes I)]\|_2\dd s \\
        &= \frac{1}{\sqrt{2^n}}\|H-H'\|_2 t.
    \end{aligned}
    \]
\end{proof}

If the initial state is a pure state randomly sampled from the Haar distribution, $\rho_0 = \ketbra{\psi_0}{\psi_0}$, following \cite{bao2023testing}, we have:
\begin{equation}
    \mathbb{E}_{|\psi_0\rangle \sim \text {Haar}_n}\left[\| U_t\ketbra{\psi_0}{\psi_0}U_t^\dagger-U_t'\ketbra{\psi_0}{\psi_0} U_t'^\dagger \|_2^2\right]=\frac{2^{n+1}}{2^n+1} \mathcal{D}(\mathcal{U}_t,\mathcal{U}_t')^2.
\end{equation}
Thus,
\begin{equation}
\label{eq:bounding_exp_frobenious_norm_with_l2_norm}
    \sqrt{\mathbb{E}_{|\psi_0\rangle \sim \text {Haar}_n}\left[\| \rho(t)-\rho'(t)\|_2^2\right]} \leq \sqrt{\frac{2}{2^n+1}}t\|H-H'\|_2\leq \sqrt{2}t\|\boldsymbol{\coef}-\boldsymbol{\coef}'\|_2.
\end{equation}

As in \eqref{eq:difference_between_expval}, we consider the difference between the expectation values of observable $O$ when measuring $\rho(t)$ and $\rho'(t)$, with the same pure initial state $\rho_0=\ketbra{\psi_0}{\psi_0}$. 
By using H\"older's inequality for Schatten norms, we have
\begin{equation}
    \abs{\expval{O}_{\rho(t)} - \expval{O}_{\rho'(t)}} = \abs{\Tr[O(\rho(t)-\rho'(t))]}\leq \|O(\rho(t)-\rho'(t)\|_1 \leq \|O\|_2\|\rho(t)-\rho'(t)\|_2.
\end{equation}
Taking the expectation value among the Haar random distributed initial state $\ket{\psi_0}$, we can obtain
\begin{equation}
\label{eq:expectation_val_l2_bound}
    \mathbb{E}_{|\psi_0\rangle \sim \text {Haar}_n}\left[ \abs{\Tr[O\rho(t)] - \Tr[O\rho'(t)]}^2\right] \leq \|O\|_2^2 \mathbb{E}_{|\psi_0\rangle \sim \text {Haar}_n}\left[|\rho(t)-\rho'(t)\|_2^2\right]\leq 2t^2\|O\|_2^2\|\boldsymbol{\coef}-\boldsymbol{\coef}'\|_2^2.
\end{equation}
The operators $O$ for which $\|O\|_2$ is small include those that are low-rank or approximately low-rank.
In the above, we have only used the second moment of $\rho_0=\ketbra{\psi_0}{\psi_0}$, and therefore we only need $\ket{\psi_0}$ to be drawn from a 2-design rather than the exact Haar distribution. 
From the above discussion we have seen that an $\ell^2$-error bound guarantees the average accuracy of predicting the expectation value of a low-rank observable for initial states drawn from a $2$-design.

\section{Discussion}

In this work we developed a protocol that efficiently learns a $k$-body Hamiltonian consisting of $M$ Pauli terms with nearly Heisenberg-limited scaling. 
Compared to previous works such as \cite{HuangTongFangSu2023learning,dutkiewicz2023advantage,bakshi2024structure}, our main contribution is that the new algorithm does not rely on any kind of geometric locality: 
the Pauli terms can involve qubits anywhere in the system, and we do not require the combined interaction strength on a qubit to be bounded, which is needed in \cite{bakshi2024structure}. 
Moreover, our protocol only requires single-qubit interleaving operations as opposed to the multi-qubit adaptive operations used in \cite{dutkiewicz2023advantage,bakshi2024structure}, and is robust to a constant amount of SPAM error. 
These features are all essential to practical implementation on current hardware.

Our results naturally lead to several open problems to be considered in future works. 
It can be readily seen that between the total evolution time lower bound for learning with the $\ell^1$-metric in Theorem~\ref{thm:lower_bound} and the upper bound in Theorem~\ref{thm:ham_learning_upper_bound}, there is a gap of order $\sqrt{M}$, ignoring logarithmic factors. 
In fact, the lower bound is tight, up to logarithmic factors, for Hamiltonians with geometrically local Pauli terms, as the learning protocol in \cite{HuangTongFangSu2023learning} provides a matching upper bound in this scenario. 
The mismatch between the current upper bound and lower bound then points to two interesting directions: can we further strengthen our algorithm to achieve better dependence on $M$, or can we strengthen the lower bound by accounting for the non-geometrically local nature of the Hamiltonian? 
If the latter can be achieved, we will be able to show rigorously that the lack of geometric locality fundamentally changes the difficulty of the task.

In this work we focused on qubit Hamiltonians, and it is natural to consider how the present approach extends to bosonic and fermionic Hamiltonians. 
The bosonic Hamiltonians, because of the infinite local degrees of freedom, may require us to change our experimental design and post-processing, and so do the fermionic Hamiltonians because of the more complicated commutation relations between terms.

Our learning protocol requires the preparation of GHZ states to achieve the desired $M$ and $n$ scaling. 
If restricted to only using product states as done in \cite{HuangTongFangSu2023learning}, can we still achieve the current scaling, or can we derive a lower bound that prohibits such scaling? 
These are all questions that should be explored in future works.

\section*{Acknowledgements}
The authors thank Anurag Anshu, Zofia Adamska, Ainesh Bakshi, Sary Bseiso, Sitan Chen, Matthew C. Ding, Zihan Hao, Hong-Ye Hu, Jun Ikeda, Liang Jiang, Akshar D. Ramkumar, Thomas Schuster, Jamie Sikora, Samson Wang, John Wright, Qi Ye, Haimeng Zhao, and Xincheng Zhang for helpful discussions. 
M.M. acknowledges funding from Caltech Summer Undergraduate Research Fellowship (SURF). 
Y.T. acknowledges funding from U.S. Department of Energy Office of Science, Office of Advanced Scientific Computing Research, DE-SC0020290.
J.P. acknowledges support from the U.S. Department of Energy Office of Science, Office of Advanced Scientific Computing Research (DE-NA0003525, DE-SC0020290), the U.S. Department of Energy, Office of Science, National Quantum Information Science Research Centers, Quantum Systems Accelerator, and the National Science Foundation (PHY-1733907). 
The Institute for Quantum Information and Matter is an NSF Physics Frontiers Center.

\appendix

\section{The Hamiltonian reshaping error bound}
\label{sec:ham_reshape_err_bound}

In this section we prove Theorem~\ref{thm:hamiltonian_reshaping}, which we restate here:
\begin{thm*}
    Let $\beta$ be a length-$n$ string of $x,y,z$, and let $\mathcal{K}_{\beta}$ be as defined in \eqref{eq:random_unitaries_basis}.
    Let $U$ be the random unitary defined in \eqref{eq:hamiltonian_reshaping}. Let $V=e^{-iH_{\mathrm{eff}}t}$ for $H_{\mathrm{eff}}$ given in \eqref{eq:effective_hamiltonian_terms}, and $t=r\tau$.
    We define the quantum channels $\mathcal{U}$ and $\mathcal{V}$ be
    \[
    \mathcal{U}(\rho) = \mathbb{E}[U\rho U^\dag],\quad \mathcal{V}(\rho) = V\rho V^\dag.
    \]
    Then 
    \[
    \|\mathcal{U}-\mathcal{V}\|_{\diamond}\leq \frac{4M^2 t^2}{r}.
    \]
\end{thm*}

\begin{proof}
    Recall from \eqref{eq:hamiltonian_reshaping} that
    \[
    U = Q_r e^{-iH\tau} Q_r\cdots Q_1 e^{-iH\tau} Q_1.
    \]
    We define
    \[
    U_j = Q_j e^{-iH\tau} Q_j,
    \]
    and the quantum channel $\mathcal{U}_j$
    \[
    \mathcal{U}_j(\rho) = \mathbb{E}_{Q_j} U_j\rho U_j^\dag.
    \]
    Because each $Q_j$ is chosen independently, we have
    \[
    \mathcal{U} = \mathcal{U}_r \mathcal{U}_{r-1}\cdots \mathcal{U}_1.
    \]
    We then define $\mathcal{V}_\tau(\rho) = e^{-iH_{\mathrm{eff}}\tau}\rho e^{iH_{\mathrm{eff}}\tau}$, which then gives us $\mathcal{V} = \mathcal{V}_\tau^r$. We will then focus on obtaining a bound for $\|\mathcal{U}_j-\mathcal{V}_\tau\|_{\diamond}$.

    Consider a quantum state $\rho$ on the current system and an auxiliary system denoted by $\alpha$, and an observable $O$ on the combined quantum system such that $\|O\|\leq 1$. We have
    \[
    \begin{aligned}
        &\Tr[O (I_\alpha\otimes U_j)\rho (I_\alpha\otimes U_j^\dag)] \\
        &= \Tr[O (I_\alpha\otimes (Q_j e^{-iH\tau}Q_j))\rho (I_\alpha\otimes (Q_j e^{iH\tau}Q_j))] \\
        &= \Tr[O\rho] - i\tau \Tr[O[I_\alpha\otimes (Q_j H Q_j),\rho]] \\
        &-\frac{\tau^2}{2}\Tr[O(I_\alpha\otimes Q_j)[I_\alpha\otimes H,[I_\alpha\otimes H, e^{-iH\tau}(I_\alpha\otimes Q_j)\rho (I_\alpha\otimes Q_j)e^{iH\tau}]](I_\alpha\otimes Q_j)],
    \end{aligned}
    \]
    where we are using Taylor's theorem with the remainder in the Lagrange form. This tells us that
    \[
    \left|\Tr[O (I_\alpha\otimes U_j)\rho (I_\alpha\otimes U_j^\dag)] - (\Tr[O\rho] - i\tau \Tr[O[I_\alpha\otimes (Q_j H Q_j),\rho]] )\right|\leq 2\|H\|^2\tau^2
    \]
    Because the two terms on the left-hand side, when we take the expectation value with respect to $Q_j$, become
    \[
    \begin{aligned}
        &\mathbb{E}_{Q_j} \Tr[O (I_\alpha\otimes U_j)\rho (I_\alpha\otimes U_j^\dag)] = \Tr[O (\mathcal{I}_\alpha\otimes \mathcal{U}_j)(\rho)] \\
        &\mathbb{E}_{Q_j} \big[\Tr[O\rho] - i\tau \Tr[O[I_\alpha\otimes (Q_j H Q_j),\rho]] \big] = \Tr[O\rho] - i\tau \Tr[O[I_\alpha\otimes H_{\mathrm{eff}},\rho]],
    \end{aligned}
    \]
    we have
    \begin{equation*}
         \left|\Tr[O (\mathcal{I}_\alpha\otimes \mathcal{U}_j)(\rho)] - (\Tr[O\rho] - i\tau \Tr[O[I_\alpha\otimes H_{\mathrm{eff}},\rho]])\right|\leq 2\|H\|^2\tau^2.
    \end{equation*}
    Because this is true for all $O$ such that $\|O\|\leq 1$,
    \begin{equation}
    \label{eq:random_unitary_short_time_approx}
        \|(\mathcal{I}_\alpha\otimes \mathcal{U}_j)(\rho)-(\rho-i\tau [I_\alpha\otimes H_{\mathrm{eff}},\rho])\|_1\leq 2\|H\|^2\tau^2.
    \end{equation}
    Similarly we can also show that
    \begin{equation}
    \label{eq:effective_ham_short_time_approx}
        \|(\mathcal{I}_\alpha\otimes \mathcal{V}_\tau)(\rho)-(\rho-i\tau [I_\alpha\otimes H_{\mathrm{eff}},\rho])\|_1\leq 2\|H_{\mathrm{eff}}\|^2\tau^2.
    \end{equation}
    Combining \eqref{eq:random_unitary_short_time_approx} and \eqref{eq:effective_ham_short_time_approx}, and by the triangle inequality, we have
    \begin{equation*}
        \|(\mathcal{I}_\alpha\otimes \mathcal{U}_j)(\rho)-(\mathcal{I}_\alpha\otimes \mathcal{V}_\tau)(\rho)\|_1\leq 2\|H\|^2\tau^2 + 2\|H_{\mathrm{eff}}\|^2\tau^2.
    \end{equation*}
    Because the above is true for all auxiliary system $\alpha$ and for all quantum states on the combined quantum system, we have
    \begin{equation}
        \|\mathcal{U}_j-\mathcal{V}_\tau\|_{\diamond}\leq 2\|H\|^2\tau^2 + 2\|H_{\mathrm{eff}}\|^2\tau^2.
    \end{equation}

    Next, we observe that
    \[
    \mathcal{U}-\mathcal{V} = \sum_{k=1}^r \mathcal{U}_r\cdots \mathcal{U}_{k+1}(\mathcal{U}_k-\mathcal{V}_\tau)\mathcal{V}_\tau^{k-1}.
    \]
    Therefore
    \[
    \|\mathcal{U}-\mathcal{V}\|_\diamond\leq \sum_{k=1}^r \|\mathcal{U}_k-\mathcal{V}_\tau\|_{\diamond}\leq r(2\|H\|^2\tau^2 + 2\|H_{\mathrm{eff}}\|^2\tau^2).
    \]
    The bound in the Theorem follows by observing that $\|H\|,\|H_{\mathrm{eff}}\|\leq M$, and $\tau = t/r$.
\end{proof}

\section{Technical lemmas for compressed sensing}
\label{sec:techinical_lemmas_compressed_sensing}
We start by introducing the orthogonal systems and the bounded orthogonal systems (BOS). Then we will move on to show that the weight-$k$ Hadamard matrix introduced in Definition~\ref{defn:weight-$k$ Hadamard matrix} is a BOS. 
After that, we will introduce the restricted isometry property, which is essential for the successful reconstruction of sparse vectors with compressed sensing. 
With that, we will show that the weight-$k$ Hadamard matrix has the desired restricted isometry property, which in turn leads to the error bound for reconstructing the sparse vector.

\begin{defn}[Orthonormal Systems]
    Let $\mathcal{D} \subset \mathbb{R}^N$ be endowed with a probability measure $\nu$. We call $\Phi=\left\{\phi_1, \ldots, \phi_N\right\}$ an orthonormal system of complex-valued functions on $\mathcal{D}$ if, for $j,k\in[N]$,
    \begin{equation}
    \label{eq: Orthonormal Systems}
        \int_{\mathcal{D}} \phi_j(\mathbf{t}) \overline{\phi_k(\mathbf{t})} d \nu(\mathbf{t})=\delta_{j, k}= \begin{cases}0 & \text { if } j \neq k, \\ 1 & \text { if } j=k ,\end{cases}
    \end{equation}
\end{defn}

Upon orthonormal systems, we can define bounded orthonormal systems.

\begin{defn}(Bounded Orthonormal Systems, \cite[Theorem~12.32]{FoucartRauhut2013})
    \label{def: BOS}
    We call $\Phi=\left\{\phi_1, \ldots, \phi_N\right\}$ a bounded orthonormal system (BOS) with constant $K$ if it satisfies Equation~\eqref{eq: Orthonormal Systems} and if
    \begin{equation}
    \label{eq: BOS}
        \left\|\phi_j\right\|_{\infty}:=\sup _{\mathbf{t} \in \mathcal{D}}\left|\phi_j(\mathbf{t})\right| \leq K \quad \text { for all } j \in[N] \text {. }
    \end{equation}
\end{defn}
The smallest value that the constant $K$ can take is $1$ since
\begin{equation}
    1=\int_{\mathcal{D}}\left|\phi_j(\mathbf{t})\right|^2 d \nu(\mathbf{t}) \leq \sup _{\mathbf{t} \in \mathcal{D}}\left|\phi_j(\mathbf{t})\right|^2 \int_{\mathcal{D}} d \nu(\mathbf{t}) \leq K^2 .
\end{equation}
When the equality $K=1$ is reached, it is necessary to have $\left|\phi_j(\mathbf{t})\right|=1$ for $\nu$. 

Consider the weight-$k$ Hadamard matrix defined in Definition~\ref{defn:weight-$k$ Hadamard matrix} with normalization factor $\frac{1}{\sqrt{2^n}}$. 
Denote the columns of this normalized weight-$k$ Hadamard matrix as $c_1,c_2,\ldots,c_D$. Since they are columns of the original Hadamard matrix, they are mutually orthogonal. 
Moreover, the BOS that consists of all columns of the original Hadamard matrix has the bounding constant $K=1$. 
Then the bounding constant for its subsystem is at most $K=1$. 
Since $K\geq 1$, we have the bounding constant $K$ for the BOS of columns of weight-$k$ Hadamard matrix is one.

We can consider functions in a BOS, of the form
\begin{equation}
\label{eq: functions on BOS}
    f(\mathbf{t})=\sum_{k=1}^N x_k \phi_k(\mathbf{t}), \quad \mathbf{t} \in \mathcal{D} .
\end{equation}
Let $\mathbf{t}_1, \ldots, \mathbf{t}_\Gamma \in \mathcal{D}$ be some sampling points and suppose we have given the sample values
\begin{equation}
    y_{\ell}=f\left(\mathbf{t}_{\ell}\right)=\sum_{k=1}^N x_k \phi_k\left(\mathbf{t}_{\ell}\right), \quad \ell \in[\Gamma].
\end{equation}
We can introduce the sampling matrix $\mathbf{A} \in \mathbb{C}^{\Gamma \times N}$ with entries
\begin{equation}
\label{eq: random sampling matrix}
    A_{\ell, k}=\phi_k\left(\mathbf{t}_{\ell}\right), \quad \ell \in[\Gamma], k \in[N].
\end{equation}
the vector $\mathbf{y}=\left[y_1, \ldots, y_\Gamma\right]^{\top}$ of sample values (measurements) can be written in the form
\begin{equation}
\label{eq: sampling function}
    \mathbf{y}=\mathbf{A} \mathbf{x}
\end{equation}
where $\mathbf{x}=\left[x_1, \ldots, x_N\right]^{\top}$ is the vector of coefficients in Equation~\eqref{eq: functions on BOS}.

Sometimes, we are interested in discrete systems, where instead of having $\mathcal{D} \subset \mathbb{R}^N$, we have $\mathcal{D} = \mathbb{F}_2^n$. 
In this case, we can consider the $\Gamma$ rows of the sampling matrix $\mathbf{A}\in\mathbb{C}^{\Gamma \times 2^n}$ to be sampled from the $2^n$ rows of the transform matrix describing this system $\mathbf{T}\in\mathcal{C}^{2^n\times 2^n}$ at random.

For a discrete BOS described by a transfer matrix $\mathbf{T}\in\mathbf{C}^{N\times N}$, we consider a random sampling matrix $\mathbf{A}\in\mathcal{C}^{\Gamma\times N}$, where $\mathbf{A}$ is a sub-matrix of $\mathbf{T}$ with its rows sampled from all rows of $\mathbf{T}$ uniformly at random. 
As in Equation~\eqref{eq: sampling function}, we define the vector $\mathbf{y}=\left[y_1, \ldots, y_\Gamma\right]^{\top}$ as the sampling result, where
\begin{equation}
    y_{\ell}=\sum_{k=1}^N  T_{\ell,k}x_k + e_\ell, \quad \ell \in[\Gamma],
\end{equation}
and $\mathbf{x}=\left[x_1, \ldots, x_N\right]^{\top}$ is the vector of coefficients to be reconstructed. 
Here, a term $e$ is added to represent noisy sampling. 
To write it in a more compact form
\begin{equation}
    \mathbf{y} = \mathbf{A}\mathbf{x} + \mathbf{e}.
\end{equation}

To reconstruct $\mathbf{x}$ with random sampling results in the presence of noise $\mathbf{e}$, robust sparse recovery methods can be used, which is introduced in Theorem~\ref{thm:compressed_sensing_main_thm}. 
These methods require the sampling matrix $\mathbf{A}$ to have the restricted isometry property (RIP).

\begin{defn}(The $M$th restricted isometry constant, \cite[Definition~6.1]{FoucartRauhut2013})
\label{def:RIP}
The $M$th restricted isometry constant $\delta_M=\delta_M(\mathbf{A})$ of a matrix $\mathbf{A}\in\mathbb{C}^{\Gamma \times N}$ is the smallest $\delta>0$ such that
\begin{equation}
    (1-\delta)\|\mathbf{x}\|_2^2 \leq\|\mathbf{A} \mathbf{x}\|_2^2 \leq(1+\delta)\|\mathbf{x}\|_2^2
\end{equation}
for all $M$-sparse vector $\mathbf{x}\in\mathbb{C}^N$. 
Equivalently, it is given by
\begin{equation}
\label{eq: equivalent_definition_of_RIC}
\delta_M=\max _{\mathbf{M} \subseteq[N], \operatorname{card}(\mathbf{M}) \leq M}\left\|\mathbf{A}_\mathbf{M}^* \mathbf{A}_\mathbf{M}-\mathbf{I d}\right\|_{2 \rightarrow 2}
\end{equation}
where $\mathbf{A}_\mathbf{M}$ is the sub-matrix of $\mathbf{A}$ with only those columns with indices in set $\mathbf{M}$.
\end{defn}
With the restricted isometry constant, we can then define the restricted isometry property.
\begin{defn}[Restricted isometry property]
 A matrix $\mathbf{A}$ is said to satisfy the restricted isometry property is $\delta_M = \delta_M(\mathbf{A})$ is small for reasonably large $M$. 
\end{defn}

To achieve sparse reconstruction from noisy samples, one way is to utilize the basic pursuit, i.e.\ doing $\ell^1$-minimization. 
To be more specific, with sparse vector $\mathbf{x}$, sampling matrix $\mathbf{A}$, sampling noise $\mathbf{e}$, and noisy sample $\mathbf{y}=\mathbf{Ax+e}$. 
With $\|e\|_2\leq \eta$, the basis pursuit reconstruct $\mathbf{x}$ from $\mathbf{y}$ and $\mathbf{A}$ by solving a $\ell^1$-minimization problem
\begin{equation}
    \label{eq: basis_pursuit}
    \underset{\mathbf{z} \in \mathbb{C}^N}{\operatorname{minimize}}\|\mathbf{z}\|_1 \quad \text { subject to }\|\mathbf{A} \mathbf{z}-\mathbf{y}\|_2 \leq \eta.
\end{equation}

Notice that $\ell^1$-minimization problem is a convex optimization problem, so the solution $\mathbf{x}^{\sharp}$ to this problem is unique and is guaranteed to be found. 
As shown in \cite[Theorem~6.11]{FoucartRauhut2013}, if $\mathbf{x}$ is a $M$-sparse vector, and $\mathbf{A}$ has the restricted isometry property with the $2M$-th restricted isometry constant $\delta_{2M}\leq \frac{77-\sqrt{1337}}{82}$, then the unique solution $\mathbf{x}^{\sharp}$ to the $\ell^1$-minimization problem in Eq~\eqref{eq: basis_pursuit} is a good approximation of $\mathbf{x}$
\begin{equation}
    \|\mathbf{x}-\mathbf{x}^{\sharp}\|_p\leq CM^{1/p-1/2} \eta,\quad 1\leq p \leq 2
\end{equation}
where $C$ is a constant only depends on $\delta_{2M}$. 
The proof of this is based on the interpretation of the restricted isometry constant in Eq~\eqref{eq: equivalent_definition_of_RIC}. $\mathbf{A}_\mathbf{M}^* \mathbf{A}_\mathbf{M}$ is close to identity, and minimizing $\|\mathbf{z}\|_1$ will guarantee finding location of the support of $\mathbf{x}$.

With basis pursuit, as long as the sampling matrix $\mathbf{A}$ associated with a BOS has sufficiently good restricted isometry property, robust sparse reconstruction with high accuracy can be guaranteed. 
In the following, we will use in \cite[Theorem~12.32]{FoucartRauhut2013} to show that the sampling matrix associated with the BOS with columns of the normalized weight-$k$ Hadamard matrix can have a sufficiently small restricted isometry constant with high probability, as long as the number of samples $\Gamma$ satisfies Eq~\eqref{eq: number of samples} in Theorem~\ref{thm:compressed_sensing_main_thm}. 
This, in turn, leads to the proof of the reconstruction error bound in Theorem~\ref{thm:compressed_sensing_main_thm}.

\begin{proof}(Proof of Theorem~\ref{thm:compressed_sensing_main_thm})
    Consider the weight-$k$ Hadamard matrix $\mathbf{H}^{(k)}$, which is associated with a BOS system with constant $K=1$. 
    The sampling matrix $\mathbf{A}\in\mathbb{C}^{\Gamma\times D}$ is a matrix whose rows are independently randomly sampled from the rows of $\mathbf{H}^{(k)}$, where $\Gamma$ is the number of samples and $D$ is the number all possible $n$-qubit, $k$-body Pauli-$Z$ strings. 
    When $k=1$ and $\Gamma$ satisfies the conditions in Eq~\eqref{eq: number of samples}, \cite[Theorem~12.32]{FoucartRauhut2013} shows that the sampling matrix $\mathbf{A}$ has the restricted isometry property with the restricted isometry constant $\delta_M\leq \eta_1 + \eta_1^2 + \eta_2$, where $\eta_1, \eta_2\in (0, 1)$ are constants that can be chosen to be small.
    With the restricted isometry property of $\mathbf{A}$, \cite[Theorem~6.11]{FoucartRauhut2013} shows that the error for robust reconstruction is bounded by the sparsity $M$ and sampling noise $\eta$, which leads to the error bound for $\mathbf{x}^{\sharp}$ in Theorem~\ref{thm:compressed_sensing_main_thm}.
\end{proof}

\section{Technical lemmas for the lower bound}
\label{sec:Technical lemmas for the lower bound}
\begin{lem}
\label{lem:diamond_distance_unitary_channels}
Consider two quantum channels $\mathcal{U}(t)$ and $\mathcal{V}(t)$ defined through
\[
\mathcal{U}(t)(\rho)=e^{-iH_1t}\rho e^{iH_1t},\quad \mathcal{V}(t)(\rho)=e^{-iH_2t}\rho e^{iH_2t}.
\]
Then 
    \begin{equation}
        \|\mathcal{U}(t)-\mathcal{V}(t)\|_{\diamond} \leq 2\|H_1-H_2\|\cdot t.
    \end{equation}
\end{lem}

\begin{proof}
    We define
    \begin{equation}
        \begin{aligned}
            U(t) = e^{-iH_1t},\quad V(t) = e^{-iH_2t}.
        \end{aligned}
    \end{equation}
    Consider
    \begin{equation}
            \frac{\dd\left(e^{-iH_1s}e^{-iH_2(t-s)}\right)}{\dd s} = -ie^{-iH_1s}\left(H_1-H_2\right)e^{-iH_2(t-s)},
    \end{equation}
    we have
    \begin{equation}
        \int_0^t \frac{\dd\left(e^{-iH_1s}e^{-iH_2(t-s)}\right)}{\dd s}\dd s = e^{-iH_1s}e^{-iH_2(t-s)}\Big|_0^t=e^{-iH_1t}-e^{-iH_2t}.
    \end{equation}
    Hence,
    \begin{equation}
        e^{-iH_1t}-e^{-iH_2t} = -i\int_0^t e^{-iH_1s}\left(H_1-H_2\right)e^{-iH_2(t-s)}ds.
    \end{equation}
    Since the diamond norm between $\mathcal{U}(t)$ and $\mathcal{V}(t)$ is bounded by the operator norm between the time evolution operators $e^{-iH_1t}$, $e^{-iH_2t}$, we have
    \begin{equation}
    \begin{split}
        \|\mathcal{U}(t)-\mathcal{V}(t)\|_{\diamond} &\leq 2\|e^{-iH_1t}-e^{-iH_2t}\|\\
        & = 2\left\|-i\int_0^t e^{-iH_1s}\left(H_1-H_2\right)e^{-iH_2(t-s)}ds\right\|\\
        &\leq 2\int_0^t \left\|H_1-H_2\right\|ds\\
        & = 2\|H_1-H_2\|t.
    \end{split}
    \end{equation}
\end{proof}

\section{The robust frequency estimation protocol}
\label{sec:robust_frequency_estimation_details}

In this section we introduce the robust frequency estimation protocol in \cite{LiTongNiGefenYing2023heisenberg}, which largely follows the idea of the robust phase estimation protocol in \cite{KimmelLowYoder2015robust}.
We modify the protocol in \cite{LiTongNiGefenYing2023heisenberg} because our goal is to obtain an accurate estimate with large probability rather than having an optimal mean-squared error scaling.
The key tool is the following lemma that allows us to incrementally refine the frequency estimate:

\begin{lem}
    \label{lem:frequency_est_refine}
    Let $\theta\in[a,b]$.
    Let $Z(t)$ be a random variable such that 
    \begin{equation}
        \label{eq:rv_correctness_condition}
        |Z(t)-e^{i\theta t}|< 1/2.
    \end{equation}
    Then we can correctly distinguish between two overlapping cases $\theta\in[a,(a+2b)/3]$ and $\theta\in[(2a+b)/3,b]$ with one sample of $Z(\pi/(b-a))$. 
\end{lem}

\begin{proof}
    These two situations can be distinguished by looking at the value of 
    \[
    f(\theta)=\sin\left(\frac{\pi}{b-a}\left(\theta-\frac{a+b}{2}\right)\right).
    \]
    We know from \eqref{eq:rv_correctness_condition} that
    \[
    \left|\Im\left(e^{-i\frac{(a+b)\pi}{2(b-a)}}Z(\pi/(b-a))\right)-f(\theta)\right|<1/2.
    \]
    where $t$ in \eqref{eq:rv_correctness_condition} is substituted by $\pi/(b-a)$ and a phase factor is added.
    
    If $$\Im\left(e^{-i\frac{(a+b)\pi}{2(b-a)}}Z(\pi/(b-a))\right)\leq 0,$$ then $f(\theta)<1/2$, which implies $\theta\in[a,(a+2b)/3]$. If $$\Im\left(e^{-i\frac{(a+b)\pi}{2(b-a)}}Z(\pi/(b-a))\right)> 0,$$ then $f(\theta)>-1/2$, which implies $\theta\in[(2a+b)/3,b]$.
\end{proof}

Using this lemma we will prove Theorem~\ref{thm:robust_frequency_estimation}, which we restate below:

\begin{thm*}
    Let $\theta\in[-A,A]$.
    Let $X(t)$ and $Y(t)$ be independent random variables satisfying
    \begin{equation}
        \label{eq:rv_correctness_condition_prob}
        \begin{aligned}
            &|X(t)-\cos(\theta t)|< 1/\sqrt{2}, \text{ with probability at least }2/3, \\
            &|Y(t)-\sin(\theta t)|< 1/\sqrt{2}, \text{ with probability at least }2/3.
        \end{aligned}
    \end{equation}
    Then with independent non-adaptive samples $X(t_1),X(t_2),\cdots,X(t_K)$ and $Y(t_1),Y(t_2),\cdots,Y(t_K)$, $t_j\geq 0$, for
    \begin{equation}
        \label{eq:number_of_samples_RPE}
        K=\mathcal{O}(\log(A/\epsilon)(\log(1/q)+\log\log(A/\epsilon))),
    \end{equation}
    \begin{equation}
    \label{eq:total_evolution_time_RPE}
    \begin{aligned}
        T=\sum_{j=1}^Kt_j=\mathcal{O}((1/\epsilon)(\log(1/q)+\log\log(A/\epsilon))),\quad \max_j t_j=\mathcal{O}(1/\epsilon),
    \end{aligned}
    \end{equation}
    we can obtain a random variable $\hat{\theta}$ such that
    \begin{equation}
        \Pr[|\hat{\theta}-\theta|>\epsilon]\leq q.
    \end{equation}
\end{thm*}

\begin{proof}
    We let $\eigen=\epsilon/3$.
    We build a random variable $S(t)$ satisfying \eqref{eq:rv_correctness_condition}, with which we will iteratively narrow down the interval $[a,b]$ containing $\theta$ until $|b-a|\leq 2\eigen$, at which point we choose $\hat{\theta}=(a+b)/2$. If $\theta\in[a,b]$ we will then ensure $|\hat{\theta}-\theta|\leq \eigen$. However, each iteration will involve some failure probability, which we will analyze later.

    To build the random variable $S(t)$, we first use $m$ independent samples of $X(t)$ and then take median $X_{\mathrm{median}}(t)$, which satisfies 
    \[
    |X_{\mathrm{median}}(t)-\cos(\theta t)|\leq 1/\sqrt{2}
    \]
    with probability at least $1-\delta/2$, where by \eqref{eq:rv_correctness_condition_prob} and the Chernoff bound 
    \[
    \delta = c_1 e^{-c_2 m},
    \]
    for some universal constant $c_1,c_2$.
    Similarly we can obtain $Y_{\mathrm{median}}(t)$ such that
    \[
    |Y_{\mathrm{median}}(t)-\cos(\theta t)|\leq 1/\sqrt{2}
    \]
    with probability at least $1-\delta/2$. With these medians we then define
    \[
    S(t) = X_{\mathrm{median}}(t) + iY_{\mathrm{median}}(t).
    \]
    This random variable satisfies
    \[
    |S(t)-e^{i\theta t}|\leq 1/2
    \]
    with probability at least $1-\delta$ using the union bound. It therefore allows us to solve the discrimination task in Lemma~\ref{lem:frequency_est_refine} with probability at least $1-\delta$.

    Whether each iteration proceeds correctly or not, the algorithm terminates after 
    \[
    L=\lceil \log_{3/2}(A/\eigen) \rceil
    \]
    iterations. In the $l$th iteration we sample $X(s_l)$ and $Y(s_l)$ where $s_l=(\pi/2A)(3/2)^{l-1}$. 
    We use $m$ samples of $X(s_l)$ and $Y(s_l)$ for computing the median in each iteration, and therefore the failure probability is at most $\delta=c_1 e^{-c_2 m}$. 
    With probability at most $L\delta$ using the union bound, one of the iteration fails. 
    In order to ensure the protocol succeeds with probability at least $1-q$, it suffices to let $\delta \leq q/L$, and it therefore suffices to choose
    \[
    m = \lceil c_3 \log(L/q) \rceil=\mathcal{O}(\log(1/q)+\log\log(A/\epsilon)).
    \]
    The total number of samples (for either $X(t)$ or $Y(t)$) is therefore as described in \eqref{eq:number_of_samples_RPE}.
    All the $t$ in each sample added together is
    \begin{equation}
        \label{eq:total_time}
        T = m\sum_{l=1}^{L}  s_l = \frac{m\pi}{2A}\sum_{l=1}^L \left(\frac{3}{2}\right)^{l-1} \leq \frac{3\pi m}{2\epsilon}=\mathcal{O}((1/\epsilon)(\log(1/q)+\log\log(A/\epsilon))),
    \end{equation} 
    thus giving us \eqref{eq:total_evolution_time_RPE}.
\end{proof}

\bibliographystyle{alpha}
\bibliography{ref.bbl}

\end{document}